%% file: 0-begin-ccpe.tex
\documentclass[final,11pt,a4paper,onecolumn,twoside,conference]{IEEEtran}

\usepackage{graphicx}

\usepackage{amssymb}
\usepackage{amsthm}
\usepackage{amsfonts} \usepackage{amsmath}

\usepackage[noadjust]{cite}

\usepackage[normal]{subfigure}
\numberwithin{equation}{section}

\usepackage{verbatim}

\newtheorem{theorem}{Theorem}

\begin{document}

\title{{\normalsize This paper is available from http://www.lcp.coppe.ufrj.br/D1HT/}\\An effective single-hop distributed hash table with high lookup performance and low traffic overhead}


\author{
\IEEEauthorblockN{Luiz Monnerat
\\luiz.monnerat@petrobras.com.br
}
\IEEEauthorblockA{Petrobras\\TIC/TIC-E\&P}
\and
\IEEEauthorblockN{Claudio Amorim
\\amorim@cos.ufrj.br
}
\IEEEauthorblockA{COPPE - Computer and Systems Engineering\\
Federal University of Rio de Janeiro (UFRJ)}
}
\markboth{L. Monnerat and C.L. Amorim}{A Low Overhead Single Hop DHT}

\maketitle

\input{10-abstract6.tex}

\begin{keywords}
Distributed hash tables; overlay networks; P2P; distributed systems; performance
\end{keywords}

\hyphenation{re-an-nouce-ments}

\input{20-introduction11.tex}
\input{25-relatedwork9-cla.tex}
\input{31-D1HT12-cla.tex}

\input{40-EDRA13-cla.tex}
\input{51-quarantine4-cla.tex}

\input{57-implementation3.tex}
\input{60-expresults4.tex}
\input{70-anresults3.tex}
\input{80-discussion4-cla.tex}

\input{300-conclusion7.tex}

\input{600-end.tex}

%% file: 10-abstract6.tex
\begin{abstract}
Distributed Hash Tables (DHTs) have been used in several applications, but most DHTs have opted to solve lookups with multiple hops, to minimize bandwidth costs while sacrificing lookup latency. This paper presents D1HT, an original DHT which has a peer-to-peer and self-organizing architecture and maximizes lookup performance with reasonable maintenance traffic, and a Quarantine mechanism to reduce overheads caused by volatile peers. We implemented both D1HT and a prominent single-hop DHT, and we performed an extensive and highly representative DHT experimental comparison, followed by complementary analytical studies. In comparison with current single-hop DHTs, our results showed that D1HT consistently had the lowest bandwidth requirements, with typical reductions of up to one order of magnitude, and that D1HT could be used even in popular Internet applications with millions of users. In addition, we ran the first latency experiments comparing DHTs to directory servers, which revealed that D1HT can achieve latencies equivalent to or better than a directory server, and confirmed its greater scalability properties. Overall, our extensive set of results allowed us to conclude that D1HT can provide a very effective solution for a broad range of environments, from large-scale corporate datacenters to widely deployed Internet applications\footnote{This is the pre-peer reviewed version of the following article: Luiz Monnerat and Claudio L. Amorim, An effective single-hop distributed hash table with high lookup performance and low traffic overhead, \emph{Concurrency and Computation: Practice and Experience (CCPE)}, 2014, which has been published in final form at http://onlinelibrary.wiley.com/doi/10.1002/cpe.3342/abstract \cite{D1HT2014CCPE}.}\textsuperscript{,}\footnote{To download this paper and other D1HT resources, including its source code, please visit the D1HT home page at http://www.lcp.coppe.ufrj.br/D1HT/}.
\end{abstract}

%% file: 20-introduction11.tex
\section{Introduction}

Distributed hash tables (DHTs) are a highly scalable solution for efficiently locating information in large-scale distributed systems; thus they have been used in a wide range of applications, from Internet games to databases. While most DHTs incur in high latencies, recent results showed that DHTs can also be applied in significant classes of applications with performance constraints, such as Internet Service Providers (ISPs), as long as they guarantee low enough latency to access information. Specifically, the development of a proprietary low-latency DHT was critical to the performance of the Amazon's Dynamo system \cite{dynamo}, where scalability, self-organization and robustness were fundamental to supporting a production system over thousands of error-prone nodes, whereas the use of central directories could lead to several problems \cite{McKusick2009}. However, the DHT implemented in Dynamo does not support open environments, has high levels of overhead and, according to its authors, it is unable to scale to very large systems, besides being very application specific. In addition, recent trends in High Performance Computing (HPC) and ISP datacenters indicate significant increases in the system sizes \cite{Barroso2008,Kindratenko2011}, including a huge demand from cloud computing \cite{Armbrust2009,Buyya2009}, which will challenge the scalability and fault tolerance of client/servers solutions. In fact, to support a wide span of large-scale distributed applications, new self-organizing DHTs with greater levels of scalability, performance and efficiency are required in order to be used as a commodity substrate for environments ranging from corporate datacenters to popular Internet applications.

The information stored in a DHT is located through \texttt{lookup} requests, which are solved with the use of \emph{routing tables} stored on all participant peers. As peers can freely enter and leave the network, DHTs typically use maintenance messages to keep the routing tables up to date. However, maintenance messages increase the DHT's network traffic, which contributes adversely to both the lookup latency and network bandwidth overheads. Overall, the size of routing tables is a critical issue in a DHT system and poses a classic latency vs. bandwidth tradeoff. Concretely, large routing tables allow faster lookups because peers will have more routing options, but they increase the bandwidth overheads due to higher maintenance traffic.

In this regard, the first DHT proposals (e.g., \cite{viceroy02,kademlia02,can01,pastry01,chord03,tapestry04}) opted to use small routing tables in such a way that each lookup takes $O(\log(n))$ hops to be solved ($n$ is the system size), aiming to save bandwidth to the detriment of latency and thus compromising the use of such \emph{multi-hop} DHTs for performance sensitive applications. However, as similar tradeoffs between latency and bandwidth occur across several technologies, the latency restrictions tend to be more critical in the long term, as it has already been shown that `over time bandwidth typically improves by more than the square of the latency reductions' \cite{patterson04}. From this perspective, a number of \emph{single-hop} DHTs have been proposed (e.g., \cite{gupta09,monneratD1HT,tang05}), which are able to provide low latency access to information because each peer maintains a full routing table. Therefore, the lookup performance achieved by these single-hop DHTs should allow their use even in latency-sensitive environments where multi-hop DHTs cannot satisfy the latency constraints. Besides, it has been shown that, for systems with high lookup rates, single-hop DHTs may in fact reduce the \emph{total} bandwidth consumption, since each lookup in a multi-hop DHT typically consumes $O(\log(n))$ more bandwidth than a single-hop lookup, and this extra lookup overhead may offset the routing table maintenance traffic \cite{rodrigues04,tang05}. Nevertheless, most single-hop DHTs still incur high bandwidth overheads, have high levels of load imbalance, or are unable to support dynamic environments.

With these problems in mind, this work provides several relevant contributions that will improve the understanding and use of single-hop DHTs in a wide range of distributed systems. We present D1HT, an original single-hop DHT combining low bandwidth overheads and good load balance even in dynamic environments, while being able to efficiently adapt to changes in the system behavior using a self-organizing and pure P2P approach. We will also present a Quarantine mechanism that can reduce the system overheads caused by volatile nodes in P2P systems.

To quantify the latencies and overheads of single-hop DHTs, we implemented D1HT and 1h-Calot \cite{tang05} from scratch and evaluated both single-hop DHTs with up to 4,000 peers and 2,000 physical nodes in two radically different environments (an HPC datacenter and a worldwide dispersed network) under distinct churn rates. Those experiments provided a number of very important results, as they validated the analyses for both DHTs, confirmed their low latency characteristics, and showed that D1HT consistently has less bandwidth requirements than 1h-Calot. Besides, our experiments also showed that D1HT has negligible CPU and memory overheads that allow its use even in heavily loaded nodes, as it used less than 0.1\% of the available CPU  cycles and very small memory to store the routing tables, even under a high rate of concurrent peer joins and leaves.

Based on the validation of the D1HT and 1h-Calot analyses, we further performed an analytical comparison among D1HT, 1h-Calot and OneHop \cite{gupta09} for system sizes of up to 10 million peers. Our results revealed that D1HT consistently had the lowest maintenance overheads, with reductions of up to one order of magnitude in relation to both OneHop and 1h-Calot. Moreover, these results also indicated that D1HT is able to support vast distributed environments with dynamics similar to those of widely deployed P2P applications, such as BitTorrent, Gnutella and KAD, with reasonable maintenance bandwidth demands. Overall, D1HT´s superior results are due to its novel P2P mechanism that groups membership changes for propagation without sacrificing latency. This mechanism was based on a theorem that will be presented in this paper, which allows each peer in a D1HT system to independently and dynamically adjust the duration of the buffering period, while assuring low latency lookups.

While scalable performance has been a fundamental argument in favor of DHTs over central directory servers, we are not aware of any published experiments demonstrating it. To fill in this gap, we performed the first experimental latency comparison among three DHTs and a directory server, using up to 4,000 peers. These experiments demonstrated the superior single-hop DHT scalability properties and provided us with other important results that will be presented in this work.

Except from our preliminary D1HT experiments \cite{monnerat09}, all previous DHT comparative evaluations with real implementations have used a few hundred \emph{physical} nodes at most and have been restricted to a single environment (e.g., \cite{gupta09,pond03,tapestry04}). Thus, the evaluation presented in this paper, which used up to 4,000 peers in two radically distinct environments, can be regarded as a highly representative experimental DHT comparison, and the first to compare the latencies provided by distinct DHTs and a directory server.

Finally, our extensive set of experimental and analytical results allowed us to conclude that D1HT consistently has the lowest overheads among the single-hop DHTs introduced so far, besides being more scalable than directory servers, and that D1HT can potentially be used in a multitude of environments ranging from HPC and ISP datacenters to applications widely deployed over the Internet.

The rest of this paper is organized as follows. The next two sections discuss related work and present the D1HT system design, and in Section \ref{sec:maintenance} we present the event dissemination mechanism used by D1HT. In Sections \ref{sec:quarantine} and \ref{sec:d1ht-implementation}, we present Quarantine and our D1HT implementation. Sections \ref{sec:expresults} and \ref{sec:anresults} present our experimental and analytical results, which are discussed in Section \ref{sec:discussion}. We then conclude the paper.

%% file: 25-relatedwork9-cla.tex
\section{Related Work} \label{sec:relatedwork}

In recent years, DHTs and P2P systems have been subjects of intense research. In particular, the design of a DHT that supports large-scale networks is a very difficult problem on its own, which poses specific challenges of scalability and efficiency. Therefore, in this work, we focus on single-hop DHTs whose event dissemination mechanisms aim at large and dynamic environments. In practice, besides D1HT, the only two other single-hop DHTs that support dynamic networks are the OneHop \cite{gupta09} and 1h-Calot \cite{tang05} systems, both of which differ from D1HT in the following fundamental ways.

The 1h-Calot \cite{tang05} DHT, which was introduced concurrently with D1HT \cite{monneratD1HT}, also uses a pure P2P topology, though they differ in significant ways. First, 1h-Calot uses event\footnote{From now on we will refer to peer joins and leaves simply as \emph{events}.} propagation trees based on peer ID intervals, while D1HT constructs its dissemination trees using message TTLs. Second, 1h-Calot uses explicit heartbeat messages to detect node failures, while D1HT relies on the maintenance messages. Third and most important, 1h-Calot peers cannot effectively buffer events and, at the same time, ensure that the lookups will be solved with a single hop, even for hypothetical systems with fixed size and peer behavior. In contrast, D1HT is able to effectively buffer events for real and dynamic systems without sacrificing latency.

Besides D1HT, OneHop is the only other single-hop DHT that is able to effectively buffer events for dissemination. However, while D1HT is a pure P2P and self-organizing system, OneHop relies on a three-level hierarchy to implement event buffering, and its topology incurs high levels of load imbalance among its different types of nodes. Additionally, to achieve its best performance, all nodes in an OneHop system must agree on some system-wide topological parameters \cite{tang05}, which are likely to be difficult to implement in practice, especially as the best parameters should change over time according to the system size and behavior.

In addition to the differences discussed above, D1HT is able to achieve overheads that are up to one order of magnitude smaller than those of both 1h-Calot and OneHop, as we will see in Section \ref{sec:anresults}.

Except for D1HT, 1h-Calot and OneHop, all other single-hop DHTs introduced so far do not support large and dynamic environments \cite{dynamo,sfdht09,risson06a,risson09,rodrigues02}. Among these, our 1h-Calot overhead results should  be also valid for SFDHT \cite{sfdht09} and 1HS \cite{risson06a}, as 1HS is based on the 1h-Calot maintenance algorithm and SFDHT uses a similar event dissemination mechanism.

While single hop DHTs must maintain full routing tables, some systems opted to use much smaller routing tables  to solve lookups with a constant number (i.e., $O(1)$) of multiple hops. For example, some $O(1)$ DHTs use $O(\sqrt{n})$ routing tables  to solve lookups with two hops. In this way, in a one million peer network, such a DHT system will maintain a routing table with a few thousands entries, which will prevent it to  address directly all one million peers with a single hop. On the other hand, these DHTs will have lower maintenance overhead, and they may be suitable for applications that are not latency sensitive. Besides solving lookups with $O(1)$ multiple hops, those systems differ from D1HT in other important aspects. For instance, Z-Ring \cite{zring05} uses Pastry \cite{pastry01} to solve lookups with two hops in systems with up to 16 million nodes. Tulip \cite{tulip05} and Kelips \cite{kelips03} use gossip to maintain routing tables of size $O(\sqrt{n})$ to solve lookups with two hops. Structured Superpeers \cite{superpeers03} and LH* \cite{litwin2005} use hierarchical topologies to solve lookups with three hops.

Accordion \cite{accordion:nsdi05} and EpiChord \cite{epichord2004} do not ensure a maximum number of lookup hops, but they use parallel lookups and adaptation techniques to minimize lookup latencies, and they can converge to one hop latencies depending on the bandwidth available. Some of those techniques can be implemented over our basic D1HT protocol (e.g., parallel lookups). Beehive \cite{beehive04} is a replication framework to speed up lookups for popular keys. Concurrently to D1HT and 1h-Calot, a mechanism for information dissemination with logarithmic trees was proposed in \cite{castro2005}, but, in contrast to D1HT, it does not perform any kind of aggregation. Scribe \cite{Castro2002} and SplitStream \cite{Castro2003} disseminate information, but they do not perform aggregation nor use logarithmic trees among several other differences in relation to D1HT.

Quarantine approaches have been proposed as a means of intervention for preventing vulnerabilities in the Internet, such as worm threats \cite{Moore03internetquarantine}, but, to the best of our knowledge, this is first work to propose, evaluate and show the effectiveness of a quarantine approach for P2P systems \cite{monneratD1HT}.

%% file: 31-D1HT12-cla.tex
\section{D1HT System Design} \label{sec:singlehopdhts}

A D1HT system is composed of a set $\mathbb{D}$ of $n$ peers and, as in Chord \cite{chord03}, the keys are mapped to peers based on \emph{consistent hashing} \cite{consistenthashing97}, where both peers and keys have IDs taken from the same identifier ring $[0:N]$, with $N >> n$. The key and peer IDs are, respectively, the hashes (e.g., SHA1 \cite{sha1}) of the key values and the peer IP addresses. Similarly to previous studies (e.g., \cite{gupta09,li04comparing,viceroy02,chord03}), we used  consistent hashing and the random properties of the cryptographic function, which allowed us to assume that the events and lookup targets are oblivious to the peer IDs and randomly distributed along the ring.

In D1HT, each peer has a full routing table, and so any lookup can be solved with just one hop, provided that its routing table is up to date. However, if the origin peer is unaware of an event that has happened in the vicinity of the target peer (e.g., a node has joined or left the system), the lookup may be initially addressed either to a wrong peer or to a peer that has already left the system. In both cases, the lookup will eventually succeed after retrying \cite{chord03}, but it will take longer than expected. To completely avoid those \emph{routing failures} (as the lookup will eventually succeed \cite{chord03}, we do consider it as a \emph{routing failure} instead of a \emph{lookup failure}), D1HT would have to immediately notify all its $n$ peers about the occurrence of any event in the system, which is simply infeasible. In practice, single-hop DHTs must try to keep the fraction of routing failures below an acceptable maximum by implementing mechanisms that can quickly notify all peers in the system about the events as they happen. These event dissemination mechanisms represent the primary distinction among the single hop DHTs, and in the next section we will present the EDRA mechanism introduced with D1HT.

As in other works (e.g., \cite{gupta09,monneratD1HT}), we will assume that the systems are churned with an event rate (or churn rate) $r$ proportional to the system size, according to Equation \ref{equ:r} below, where $S_{avg}$ is the peer average session length.
\begin{equation} \label{equ:r} r = 2 \cdot n /S_{avg} \end{equation}

We refer to the \emph{session length} as the amount of time between a peer's join and its subsequent leave; thus, Equation \ref{equ:r} simply assumes that, as expected, each peer generates two events per session (one join and one leave). As the average session lengths of a number of different P2P systems have already been measured (e.g., \cite{bellissimo04,saroiu02,steiner09}), the equation above allows us to calculate event rates that are representative of widely deployed applications. In Sections \ref{sec:expresults} and \ref{sec:anresults}, we will present experimental and analytical results with different session lengths, which will allow us to evaluate its effect on the maintenance overheads.

In D1HT, any message should be acknowledged to allow for retransmissions in the case of failures, which can be done implicitly by a protocol like TCP or be explicitly implemented over an unreliable protocol like UDP. We assume that the maintenance messages are transmitted over UDP to save bandwidth, but we consider the choice of the transport protocol for all other messages as an implementation issue. We also consider that the details of the joining protocol should be decided at the implementation level. In Section \ref{sec:d1ht-implementation}, we will discuss how we ensure message delivery in our D1HT implementation and what joining protocol we used.

D1HT is a pure P2P and self-organizing system, but its flat topology does not prevent it from being used as a component of hierarchical approaches aiming to exploit the heterogeneity of the participant nodes in a system. For example, the FastTrack network \cite{liang06} has two classes of nodes: the super nodes (SN) and ordinary nodes (ON). SNs are better provisioned nodes, and each SN acts as a central directory for a number of ONs, while flooding is used among the SNs. As measurements \cite{liang06} have shown that FastTrack should have less than 40K SNs with an average session length of 2.5 hours, the analysis that we will present in Section \ref{sec:maintenance} shows that we could use a D1HT system to connect the SNs with maintenance costs as low as 0.9 kbps per SN. This overhead should be negligible, especially if we consider that the SNs are well-provisioned nodes and that we would avoid the flooding overheads while improving the lookup performance.

We will not address issues related to malicious nodes and network attacks, although it is clear that, due to their high out-degree, single-hop DHTs are naturally less vulnerable to those kinds of menaces than low-degree multi-hop DHTs.

%% file: 40-EDRA13-cla.tex
\section{EDRA} \label{sec:maintenance} \label{sec:d1ht}

As each peer in a D1HT system should know the IP address of every other peer, any event should be \emph{acknowledged} by all peers in the system in a timely fashion to avoid stale routing table entries. Here, we say that a peer \emph{acknowledges} an event when it either detects the join (or leave) of its predecessor or receives a message notifying of an event.

To efficiently propagate any event to all peers in a system, D1HT makes use of the Event Detection and Report Algorithm (EDRA), which can announce any event to the whole system in logarithmic time with a pure P2P topology and provides good load-balance properties coupled with low bandwidth overhead. Additionally, EDRA is able to group several events into a single message to save bandwidth, yet it ensures an upper bound on the fraction of stale routing table entries.

At first glance, grouping several event notifications per message seems to be an obvious and easy way to save bandwidth, as any peer can locally buffer the events that occur during a period of time and forward them in a single message. However, such a mechanism imposes delays in the event dissemination, which in turn will lead to more stale entries in the routing tables; thus, the difficult question is the following: \emph{For how long can each peer buffer events while assuring that the vast majority of the lookups (e.g., $99\%$) will be solved with just one hop?} This problem is especially difficult because the answer depends on a number of factors that vary unpredictably, including the system size and churn rate. EDRA addresses this issue based on a theorem that will be presented in this section, which allows each peer to independently adjust the length of the buffering period while assuring that at least a fraction 1-$f$ of the lookups will be solved with a single hop ($f$ is typically $1\%$, but it can be tuned according to the application).

In this section, we will formally describe EDRA by means of a set of rules, prove its correctness and load balance properties, and present its analysis. Before we begin, we will define a few functions to make the presentation clearer. For any $i \in \mathbb{N}$ and $p \in \mathbb{D}$, the $i_{th}$ successor of $p$ is given by the function $succ(p,i)$, where $succ(p,0)$=$p$ and $succ(p,i)$ is the successor of $succ(p,i$-$1)$ for $i>0$. Note that for $i \ge n$, $succ(p,i)$=$succ(p,i$-$n)$. In the same way, the $i_{th}$ predecessor of a peer $p$ is given by the function $pred(p,i)$, where $pred(p,0)$=$p$ and $pred(p,i)$ is the predecessor of $pred(p,i$-$1)$, for $i>0$. As in \cite{viceroy02}, for any $p \in \mathbb{D}$ and $k \in \mathbb{N}$, $stretch(p,k)$=$\{\forall p_i \in \mathbb{D} \mid p_i$=$succ(p,i) ~ \land ~ 0 \le i \le k\}$. Note that $stretch(p,n$-$1)$=$\mathbb{D}$ for any $p \in \mathbb{D}$.

\subsection{The EDRA Rules} \label{sec:rules}

In this section, we will first present a brief description of EDRA and then formally define it. To save bandwidth, each peer buffers the events acknowledged during intervals of $\Theta$ seconds ($\Theta$ intervals), where $\Theta$ is dynamically tuned (as it will be seen in Section \ref{sec:tuning}). At the end of a $\Theta$ interval, each peer propagates the events locally buffered by sending up to $\rho$=$\lceil\log_2(n)\rceil$ maintenance messages, as shown in Figure \ref{fig:hops}. Each maintenance message $M(l)$ will have a Time-To-Live (TTL) counter $l$ in the range [0:$\rho$) and will be addressed to $succ(p,2^l)$. To perform event aggregation while assuring that any event will reach all peers in the system, each message $M(l)$ will include all events brought by any message $M(j), j>l$, received in the preceding $\Theta$ seconds. To initiate an event dissemination, the successor of the peer suffering the event will include it in all messages sent at the end of the current $\Theta$ interval.
The rules below formally define the EDRA algorithm we have briefly described above:

\begin{description}
\item[Rule 1:] ~ Every peer will send at least one and up to $\rho$ maintenance messages at the end of each $\Theta$ sec interval ($\Theta$ interval), where $\rho$=$\lceil\log_2(n)\rceil$.

\item[Rule 2:] ~ Each maintenance message $M(l)$ will have a distinct TTL $l$, $0 \le l < \rho$, and carry a number of events. All events brought by a message $M(l)$ will be \emph{acknowledged} with $TTL$=$l$ by the receiving peer.

\item[Rule 3:] ~ A message will only contain events acknowledged during the ending $\Theta$ interval. An event acknowledged with $TTL$=$l$, $l>0$, will be included in all messages with $TTL<l$ sent at the end of the current $\Theta$ interval. Events acknowledged with $TTL$=0 will not be included in any message.

\item[Rule 4:] ~ Messages with $TTL$=0 will be sent even if there is no event to report. Messages with $TTL>0$ will only be sent if there are events to be reported.

\item[Rule 5:] ~ If a peer $P$ does not receive any message from its predecessor $p$ for $T_{detect}$ sec, $P$ will probe $p$ to ensure that it has left the system and, after confirmation, $P$ will acknowledge $p$ leaving.

\item[Rule 6:] ~ When a peer detects an event in its predecessor (it has joined or left the system), this event is considered to have been \emph{acknowledged} with $TTL$=$\rho$ (so it is reported through $\rho$ messages according to Rule 3).

\item[Rule 7:] ~ A peer $p$ will send all messages with $TTL$=$l$ to $succ(p,2^l)$.

\item[Rule 8:] ~ Before sending a message to $succ(p,k)$, $p$ will discharge all events related to any peer in $stretch(p,k)$.

\end{description}

\begin{figure*}[tb] \centerline{\includegraphics[width=14.0cm]{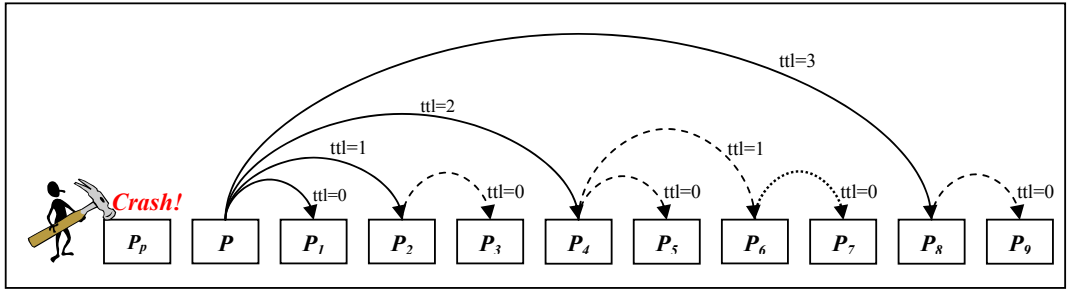}}
\caption{This figure shows a D1HT system with 11 peers, where peer $p$ crashes and this event is detected and reported by its successor $P$. In the figure, peers $P_i$ are such that $P_i$=$succ(P,i)$. The figure also shows the $TTL$ of each message sent.}\label{fig:hops}
\end{figure*}

Rules 4 and 5 should allow each peer to maintain pointers to its correct successor and predecessor even in the case of peer failures. Moreover, to improve robustness, any peer $p$ should run a local stabilization routine whenever it does not receive a reply to a msg with TTL=0 or when it receives a msg with TTL=0 (or TTL=1) from others than its predecessor (or $pred(p,1)$), and this routine should allow any peer to detect its correct predecessor and successor even if multiple consecutive peers fail simultaneously. As there are already routines proposed in the literature that can accomplish these tasks (e.g. \cite{chord03}), we leave its details to be decided at the implementation level.

Figure \ref{fig:hops} shows how EDRA disseminates information about one event to all peers in a system according to the rules just presented, and it illustrates some properties that we will formally prove in the next section. The figure presents a D1HT system with 11 peers ($\rho=4$), where peer $p$ crashes and this event $\varepsilon$ is detected and reported by its successor $P$. The peers in the figure are shown in a line instead of a ring to facilitate the presentation. Note that $P$ acknowledges $\varepsilon$ after $T_{detect}$ sec (Rule 5) with $TTL=\rho$ (Rule 6). According to Rules 3 and 7, $P$ will forward $\varepsilon$ with $\rho=4$ messages addressed to $P_1$=$succ(P,2^0)$, $P_2$=$succ(P,2^1)$, $P_4$=$succ(P,2^2)$ and $P_8$=$succ(P,2^3)$, as represented by the solid arrows in the figure. As $P_2$, $P_4$ and $P_8$ will acknowledge $\varepsilon$ with $TTL>0$, they will forward it to $P_3$=$succ(P_2,2^0)$, $P_5$=$succ(P_4,2^0)$, $P_6$=$succ(P_4,2^1)$ and $P_9$=$succ(P_8,2^0)$, as represented by the dashed arrows. Because $P_6$ will acknowledge $\varepsilon$ with $TTL$=1, it will further forward it to $P_7$=$succ(P_6,2^0)$ (doted arrow). Note that Rule 8 prevents $P_8$ from forwarding $\varepsilon$ to $succ(P_8,2^1)$ and $succ(P_8,2^2)$, which in fact are $P$ and $P_3$, saving these two peers from having to acknowledge $\varepsilon$ twice.

\subsection{EDRA Correctness} \label{sec:correcteness}

The EDRA rules ensure that any event will be delivered to all peers in a D1HT system in logarithmic time, as we will shortly show in Theorem \ref{theo:correct}. For this theorem, we will ignore message delays and consider that all peers have synchronous intervals, i.e., the $\Theta$ intervals of all peers start at exactly the same time. The absence of message delays means that any message will arrive immediately at its destination, and because we are also considering synchronous $\Theta$ intervals, any message sent at the end of a $\Theta$ interval will arrive at its destination at the beginning of the subsequent $\Theta$ interval. We will also assume that no new event happens until all peers are notified about the previous event. All these practical issues will be addressed in Section \ref{sec:practical}.
\begin{theorem} \label{theo:correct}An event $\varepsilon$ that is acknowledged by a peer $p$ with $TTL$=$l$ and by no other peers in $\mathbb{D}$ will be forwarded by $p$ through $l$ messages in such a way that $\varepsilon$ will be acknowledged exactly once by all peers in $stretch(p,2^l$-$1)$ and by no other peer in the system. The average time $T_{sync}$ for a peer in $stretch(p,2^l$-$1)$ to acknowledge $\varepsilon$ will be at most $l\cdot \Theta/2$ after $p$ acknowledged $\varepsilon$. \end{theorem}
\begin{proof}By strong induction in $l$. For $l$=1, the EDRA rules imply that $p$ will only forward $\varepsilon$ through a message with $TTL$=0 to $succ(p,1)$. As this message should be sent at the end of the current $\Theta$ interval, $succ(p,1)$ will acknowledge $\varepsilon$ at most $\Theta$ sec after $p$ acknowledged it, making the average time for peers in $stretch(p,1)$=$\{p,succ(p,1)\}$ to be $T_{sync}$=$ (\Theta+0)/2 $=$ \Theta/2$ (at most). Thus, the claim holds for $l$=1.

For $l>1$, the EDRA rules imply that $p$ will forward $\varepsilon$ through $l$ messages at the end of the current $\Theta$ interval, each one with a distinct TTL in the range [0~,~$l$). Then, after $\Theta$ sec (at most) each peer $p_k$=$succ(p,2^k)$, $0\le k < l$, will have acknowledged $\varepsilon$ with $TTL$=$k$. Applying the induction hypothesis to each of those $l$ acknowledgements, we deduce that each acknowledgment made by a peer $p_k$ implies that all peers in $stretch(p_k,2^k$-$1)$ will acknowledge $\varepsilon$ exactly once. Accounting for all $l$-$1$ acknowledgments made by the peers $p_k$, and that Rule 8 will prevent $\varepsilon$ from being acknowledged twice by any peer in $stretch(p,2^\rho$-$n)$, we conclude that $\varepsilon$ will be acknowledged exactly once by all peers in $stretch(p,2^l$-$1)$. By the induction hypothesis, none of those peers will forward $\varepsilon$ to a peer outside this range, so $\varepsilon$ will not be acknowledged by any other peers in the system. The induction hypothesis also ensures that the average time for the peers in each $stretch(p_k,2^k$-$1)$ to acknowledge $\varepsilon$ will be (at most) $k \cdot \Theta/2$ after the respective peer $p_k$ acknowledged it, which will lead to $T_{sync}$=$l \cdot \Theta/2$ (at most) for $stretch(p,2^l$-$1)$.

\end{proof}

\vspace{5pt}

Applying Theorem \ref{theo:correct} and the EDRA rules to a peer join (or leave) that is acknowledged by its successor $p$, we can conclude that this event will be further acknowledged exactly once by all peers in $stretch(p,2^{\rho}$-$1)$=$\mathbb{D}$. Moreover, the upper bound on the average acknowledge time will be $\rho \cdot \Theta/2$. We can thus formally ensure three very important EDRA properties. First, any event will be announced to all peers in a D1HT system, ensuring that they will receive the necessary information to maintain their routing tables. Second, each peer will be notified of any event just once, avoiding unnecessary bandwidth overheads and ensuring good income load balance. Third, for each event, the average notification time is bounded by $\rho \cdot \Theta/2$, and this result will be used in Section \ref{sec:tuning} to develop a mechanism that will allow each peer in a D1HT system to dynamically find the optimal value for $\Theta$ based on the current system size and behavior.

We can also show that the last peer to acknowledge an event would be $succ(p,n-1)$ (which is $pred(p,0)$), $\rho \cdot \Theta$ secs after $p$ had acknowledged the event. In practice, $pred(p,0)$ will know about the event much before, due to the stabilization routine discussed in Section \ref{sec:rules}.

\subsection{Practical Aspects} \label{sec:practical}

In Theorem \ref{theo:correct}, we did not consider the effects of message delays and asynchronous $\Theta$ intervals; thus, we will turn to them in this section. To compute those effects, we will assume that each maintenance message will require an average delay of $\delta_{avg}$ to reach its target, and it will typically arrive at the middle of a $\Theta$ interval. Therefore, under those more realistic assumptions, each peer in the event dissemination path should add an average of $\delta_{avg}$+$\Theta/2$ to the event propagation time, leading to the adjusted value $\rho \cdot (2\cdot\delta_{avg}+\Theta)/4$. Note that we have not yet considered the time to detect the event, which we will assume to be $T_{detect}$=$2 \cdot \Theta$, reflecting the worst-case scenario in which, after one missing message with $TTL$=0, a peer will probe its predecessor for up to $\Theta$ sec before reporting its failure. Thus, the upper bound on the average acknowledge time for any event will be
\begin{equation}
\label{equ:Tavg}
{T_{avg}= 2 \cdot \Theta + \rho \cdot (\Theta+ 2 \cdot \delta_{avg})/4 ~ sec}
\end{equation}
Equation \ref{equ:Tavg} overestimates $T_{avg}$ because it only considers the worst-case of peer failures, whereas we should have set $T_{detect}=0$ for joins and voluntary leaves.

In Theorem \ref{theo:correct}, we also considered that no new event would happen until all peers had been notified of a previous event, which is not a reasonable assumption for real and dynamic systems. While the admission of new peers should be correctly handled by the joining protocol, peer leaves are more complicated, and we may not expect that all peers in a system will have identical routing tables. For instance, when a peer fails before forwarding the locally buffered events, the propagation chain for these events will be partially broken. However, because this problem may occur only once per peer session (at most), it should not have a significant effect, as the duration of the buffering period (a few tens of seconds at most \cite{monneratD1HT}) is typically orders of magnitude smaller than the average session length (e.g., almost three hours for KAD and Gnutella). For example, for systems with Gnutella behavior the results presented in \cite{monneratD1HT} show that this problem should happen only once for about 1500 $\Theta$ intervals. So even if all nodes leaves were due to failures at the exact end of the $\Theta$ intervals (which is a very conservative assumption), less than 0.07\% of the events forwarded by each peer during its lifetime would be lost due to this reason. If we consider that half of the leaves are due to failures (which is also conservative), and that those failures typically occurs at the middle of the $\Theta$ intervals, then less than 0.02\% of the events forwarded by each peer would be lost (in other words, only one in around 6000 events forwarded by each peer would get lost due to this issue).

In fact, in Section \ref{sec:expresults}, we will see that D1HT was able to solve more than 99\% of the lookups with just one hop in all experiments, even under a high rate of concurrent joins and leaves, which is a strong experimental evidence that the routing failures due to those practical issues should not be relevant in relation to $f$.

Besides the problems discussed so far, there are a number of practical situations that can lead to stale routing table entries in D1HT and other DHT systems, and we will not be able to completely remedy all of them. For that reason, as in many other systems (e.g., \cite{gupta09,kelips03,epichord2004,accordion:nsdi05,kademlia02}), any D1HT implementation should allow the peers to learn from the lookups and maintenance messages to perform additional routing table maintenance without extra overhead. For example, a message received from an unknown peer should imply its insertion in the routing table. In the same way, routing failures will provide information about peers that have left or joined the system. In addition, many other known mechanisms that are commonly used in other DHT systems could be implemented on top of our base D1HT design, such as event re-an\-nounce\-ments \cite{tang05} and gossip \cite{risson09} to improve routing table accuracy, or parallel lookups (as in \cite{accordion:nsdi05,epichord2004}) to mitigate the latency penalties caused by timeouts due to missed leave notifications. We should point out that even with parallel lookups, the D1HT lookup bandwidth demands should be smaller than those of multi-hop DHTs for large systems.

\subsection{Tuning EDRA} \label{sec:tuning}

In this section, we will show how to tune EDRA to ensure that a given fraction 1-$f$ of the lookups will be solved with one hop, where $f$ can be statically defined (e.g., $f$=1\%) or dynamically adjusted.

As the lookups are solved with just one hop, to achieve $f$ it is enough to ensure that the hops will fail with probability $f$, at most. As discussed in Section \ref{sec:singlehopdhts}, we may assume that the lookup targets are random, as in many other studies (e.g., \cite{gupta09,li04comparing,superpeers03,chord03}). Then, the average fraction of routing failures will be a direct result of the number of stale routing table entries. In that manner, to satisfy $f$, it suffices to assure that the average fraction of stale routing table entries is kept below $f$ \cite{gupta09}.

Given that the average acknowledge time is at most $T_{avg}$, the average number of stale routing table entries will be bounded by the numbers of events occurred in the last $T_{avg}$ seconds, i.e., $T_{avg} \cdot r$. Then, we should satisfy the inequality $T_{avg} \cdot r / n \leq f$, and thus, by Equations \ref{equ:r} and \ref{equ:Tavg}, the maximum $\Theta$ value should be
\begin{equation} \label{equ:Theta}
 \Theta = (2 \cdot f \cdot S_{avg} - 2 \cdot \rho \cdot \delta_{avg})/(8+ \rho) ~ sec .
\end{equation}
The equation above requires each peer to know the average message delay; to ease the implementation, we will simply assume that $\delta_{avg}$=$\Theta/4$, which is an overestimation according to previously published results \cite{monneratD1HT,saroiu02}. Then
\begin{equation} \label{equ:ThetaN}
{\Theta = (4 \cdot f \cdot S_{avg})/(16+3 \cdot \rho) ~ sec }.
\end{equation}
As all D1HT peers know about any event in the system, Equations \ref{equ:r} and \ref{equ:ThetaN} allow each peer to dynamically calculate the optimal value for $\Theta$ based on the event rate that is observed locally, without the need for further communication or agreement with other peers. This allows each peer in a D1HT system to independently adapt to the environment dynamics to maximize the buffering period without penalizing latency, even for large real systems whose size and peer behavior typically change over time. In contrast, as discussed in Section \ref{sec:relatedwork}, peers in all other P2P single hop DHTs are unable to independently calculate the length of event buffering periods, even for hypothetical systems with fixed size and peer behavior.

To make D1HT more robust to sudden bursts of events, we extended the original D1HT analysis to allow each peer to overestimate the maximum number of events it may buffer ($E$) according to Equation \ref{equ:e} below. This equation was derived from Equation \ref{equ:ThetaN} with the assumption that peers in a D1HT system observe similar event rates (which led us to assume that $r$=$E/\Theta$).
\begin{equation} \label{equ:e}
{E = (8 \cdot f \cdot n)/(16+3 \cdot \rho) ~~~events }
\end{equation}

\subsection{Maintenance Traffic and Load Balance} \label{sec:loadbalanceandperformance}

While we have proven that EDRA ensures a good income load balance, it does not seem at first glance to provide good balance in terms of outgoing traffic. For instance, in Figure \ref{fig:hops}, peer $P$ sent four messages reporting $p$ crash, while $P_1$ did not send a single message. But we should not be concerned with the particular load that is generated by a single event, as it should not exceed a few bytes per peer. Nevertheless, we must guarantee good balance with respect to the aggregate traffic that is necessary to disseminate information about all events as they happen, and this outgoing maintenance load balance will rely on the random distribution properties of the hash function used. As discussed in Section \ref{sec:singlehopdhts}, the chosen function is expected to distribute the peer IDs randomly along the ring. Then, as in many other studies (e.g., \cite{gupta09,li04comparing,viceroy02,chord03}), we will assume that the events are oblivious to the peer IDs, leading to a randomly distributed event rate $r$. Thus, the average number of messages each peer sends per $\Theta$ interval will be (including message acknowledgments)

\begin{equation} \label{equ:income}
{( N_{msgs} \cdot ( v_m + v_a) + r \cdot m \cdot \Theta) /\Theta ~ \textrm{ bit/sec} }
\end{equation}
where $m$ is the number of bits to describe an event, and $v_m$ and $v_a$ are the bit overheads (i.e., headers) per maintenance message and per message acknowledgment, respectively. As no peer will exchange maintenance messages with any node outside $\mathbb{D}$, Equation \ref{equ:income} will reflect both the incoming and outgoing average maintenance traffic.

\subsection{Number of Messages}

Equation \ref{equ:income} requires us to determine the average number of messages a peer sends, which is exactly the purpose of the following theorem.
\begin{theorem} \label{theo:numberofmsgs} The set of peers $S$ for which a generic peer $p$ acknowledges events with $TTL \ge l$ satisfies $|S|$=$2^{\rho-1}$. \end{theorem}
\begin{proof}By induction on $j$, where $j$=$\rho$-$l$. For $j$=0, Rule 2 ensures that there is no message with $TTL\ge l$=$\rho$. Then, the only events that $p$ acknowledges with $TTL \ge \rho$ are those related to its predecessor (Rule 6), so $S$=$\{pred(p,1)\}$, which leads to $|S|$=$1$=$2^0$=$2^{\rho-l}$.

For $j>0$, $l=\rho$-$j<\rho$. As $S$ is the set of peers for which $p$ acknowledges events with $TTL \ge l$, we can say that $S $=$ S1 \cup S2$, where $S1$ and $S2$ are the sets of peers for which $p$ acknowledges events with $TTL$=$l$ and $TTL>l$, respectively. By the induction hypothesis, $|S2|$=$2^{\rho-(l+1)}$. As $l<\rho$, the predecessor $p$ will not be in $S1$ (Rule 6). Thus, as Rule 7 implies that $p$ only receives messages with $TTL$=$l$ from a peer $k$, where $k$=$pred(p,2^l)$, we have that $S1$ will be the set of peers for which $k$ forwards events through messages with $TTL$=$l$. By Rule 3, $S1$ is the set of peers for which $k$ acknowledges events with $TTL>l$, and as the induction hypothesis also applies to the peer $k$, it follows that $|S1|$=$2^{\rho-(l+1)}$. By Theorem \ref{theo:correct}, we know that any peer $p$ acknowledges each event only once, ensuring that $S1$ and $S2$ are disjoint, and thus, $|S|$=$|S1|+|S2|$=$2^{\rho-(l+1)}+2^{\rho-(l+1)}$=$2^{\rho-l}$.
\end{proof}

\vspace{5pt}

The EDRA Rules 3 and 4 ensure that a peer $p$ will only send a message with $TTL=l>0$ if it acknowledges at least one event with $TTL \ge l+1$. Then, based on Theorem \ref{theo:numberofmsgs}, we can state that $p$ will only send a message with $TTL=l>0$ if at least one in a set of $2^{\rho-l-1}$ peers suffers an event. As the probability of a generic peer suffering an event in a $\Theta$ interval is $\Theta \cdot r/n$, the probability $P(l)$ of a generic peer sending a message with $TTL=l \ne 0$ at the end of each $\Theta$ interval is
\vspace{-5pt}
{\begin{equation} \label{equ:messageprobability}
{ P(l) = 1- (1 - 2 \cdot r \cdot \Theta / n)^k, \textrm{ where~} k=2^{\rho-l-1} }.
\end{equation} }
\noindent As the messages with $TTL$=0 are always sent, the average number of messages sent by each peer per $\Theta$ interval will be
\begin{equation} \label{equ:numberofmsgs}
{ N_{msgs} = 1 + \sum_{l=1}^{\rho-1} P(l) }.
\end{equation}

%% file: 51-quarantine4-cla.tex
\section{Quarantine} \label{sec:quarantine}

In any DHT system, peer joins are costly, as the joining peer has to collect information about its keys and the IP addresses to fill in its routing table, and this joining overhead may be useless if the peer departs quickly from the system. While ideally all peers in a DHT system should be able to solve lookups with a single hop at any time, in extremely large and dynamic systems the overheads caused by the most volatile peers can be excessive. Moreover, P2P measurement studies \cite{chu02,saroiu02,steiner09} have shown that the statistical distributions of session lengths are usually heavy tailed, which means that peers that have been connected to the system for a long time are likely to remain alive longer than newly arrived peers. To address those issues, we proposed a \emph{Quarantine} mechanism, in which a joining peer will not be immediately allowed to take part in the D1HT overlay network, but it will be able to perform lookups at any moment. In this way, the most volatile peers will cause insignificant overheads to the system, while the other peers will be able to solve lookups with just one hop most of the time (typically, during more than 95\% of their session lengths).

To join a D1HT system, a peer $p$ retrieves the keys and IP addresses from a set of peers $\mathbb{S}$ (which can include just one peer, e.g., the successor of the joining peer). With Quarantine, the peers in $\mathbb{S}$ will wait for a Quarantine period $T_q$ (which can be fixed or dynamically tuned) before sending the keys and IP addresses to $p$, postponing its insertion into the D1HT ring. While $p$ is in Quarantine, its join will not be reported, and it will not be responsible for any key. The overhead reductions attained can be analytically quantified based on the Quarantine period and the statistical distribution of the session lengths, as in a system with $n$ peers, only the $q$ peers with sessions longer than $T_q$ will effectively take part of the overlay network and have their events reported.

To be able to perform lookups during its Quarantine, a quarantined peer $p$ will choose the nearest (in terms of latency) and best provisioned peers from $\mathbb{S}$ and will forward its lookups to those \emph{gateway} peers. To avoid excessive loads, each gateway peer may limit the rate of lookups it will solve on behalf of quarantined peers, even though the experimental results that we will show in Section \ref{sec:exp_hpc_band}, where each D1HT peer used less than 0.1\% of the available CPU cycles, indicate that the load imposed on the gateway peers should not be high. Anyway, this extra load should be much inferior than those handled by superpeers (or supernodes) in hierarchical systems like FastTrack \cite{liang06}, OneHop \cite{gupta09} or Structured Superpeers \cite{superpeers03}.

With the Quarantine mechanism, we avoid the join and leave overheads for peers with session lengths smaller than $T_q$, but newly incoming peers will have their lookups solved in two hops while they are in Quarantine. We believe that this extra hop penalty should be acceptable for several reasons. First, the additional hop should have low latency, as it will be addressed to a nearby peer. Second, this extra overhead will only be necessary during a short period (e.g., less than 6\% of the average session length). Third, the Quarantine mechanism should have beneficial effects even to the volatile and gateway peers, as they will not incur the overhead of transferring the keys and routing tables. Fourth, the Quarantine mechanism should significantly reduce the maintenance overheads of all peers in the system (as will be confirmed by the results presented in Section \ref{sec:anresults}).

Furthermore, the Quarantine mechanism can also be used for other purposes. For instance, Quarantine can be used to minimize sudden overheads due to flash crowds, by increasing $T_q$ whenever the event rate reaches the upper limit that can be comfortably handled by the system.

%% file: 57-implementation3.tex
\section{D1HT Implementation} \label{sec:d1ht-implementation}

We implemented D1HT from scratch, resulting in more than 8.000 lines of dense C++ code, even though we did not yet implement the Quarantine mechanism. This implementation is fully functional and was tested on thousands of nodes running Linux, and its source code is freely available \cite{D1HTsource}.

Our implementation uses a variant of the Chord joining protocol \cite{chord03}, with a few important differences. First, any join is announced to the whole system by EDRA. Second, the new peer $p$ gets the routing table from its successor $p_s$. Third, to prevent $p$ from missing events while its joining is notified to the system, $p_s$ will forward to $p$ any event it knows until $p$ receives messages with all different TTLs.

\input{fig-headers.tex}

To save bandwidth and minimize latency, the maintenance and lookup messages are sent with UDP, and TCP is used for all other types of communications (routing table transfers, stabilizations, etc.). Each D1HT instance has a default IPv4 port, but any peer may choose an alternative port when joining the system. Thus, we expect that most events will be identified only by the peer's four byte IPv4 address (as most peers should use the default port), which led us to propose the message header layout as shown in Figure \ref{fig:headers}. Then, for Equation \ref{equ:income}, we expect that the average $m$ value will be around 32 bits.

Each D1HT peer stores its routing table as a local hash table indexed by the peer IDs in such a way that any peer needs only to store the IPv4 addresses of the participant peers (including the port number), leading to a memory overhead of about 6$n$ bytes in each peer (plus some additional space to treat eventual hash collisions). In this way, for environments such as HPC and ISP datacenters, each routing table will require a few hundred KBs at most. For a huge one million Internet wide D1HT deployment, each routing table would require around 6 MB, which is negligible for domestic PCs and acceptable even for small devices, such as cell phones and media players.

%% file: fig-headers.tex
\begin{figure*}[!t]
\centering
\includegraphics[width=10.8cm]{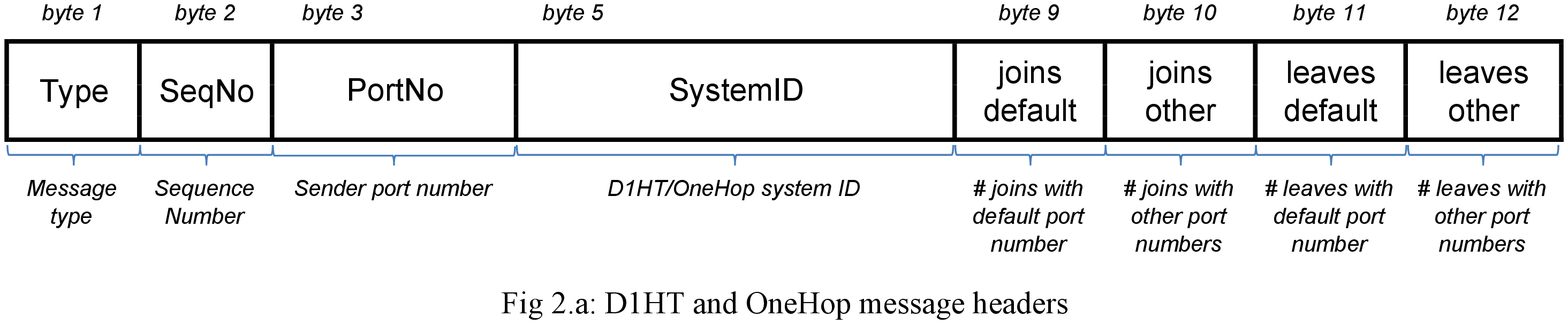}
\\
\vspace{-4.6cm}
\includegraphics[width=14.2cm]{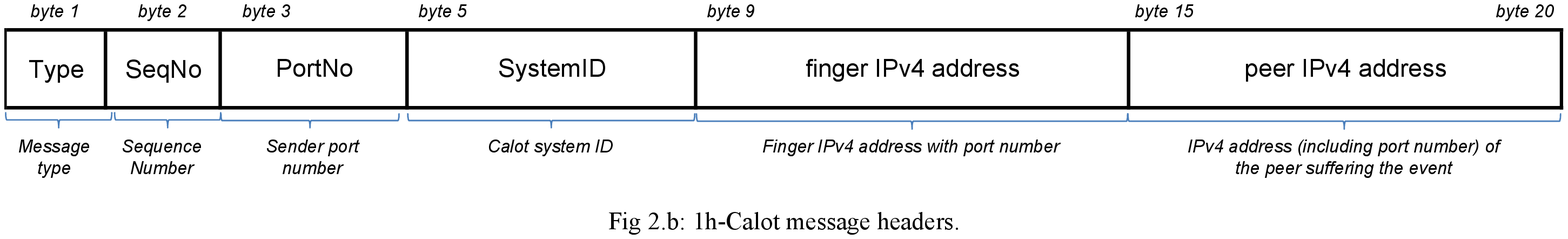}
\vspace{-8.0cm}
\\
\caption{Message headers used in our implementations and analyses. The SeqNo field is necessary to assure message delivery over UDP, and the SystemID field allows any peer to discard unsolicited messages received from other DHT systems. Each 1h-Calot maintenance message has a fixed size of 48 bytes ($v_c$=384 bits, including 28 bytes for the IPv4 and UDP headers). Each D1HT and OneHop message has a fixed part with 40 bytes ($v_m$=320 bits, including IPv4 and UDP headers), followed by the IPv4 addresses (without port numbers) of the peers that have joined and left the system in the default port ($m$=32 bits), and the IPv4 addresses (with port numbers) of the peers that have joined and left the system using others ports ($m$=48 bits). All acknowledgment  and heartbeat messages for the three systems have just the four first fields shown (Type, SeqNo, PortNo and SystemID), and so $v_a$=$v_h$=288 bits (including IPv4 and UDP headers).}
\label{fig:headers}
\vspace{-10pt}
\end{figure*} 

%% file: 60-expresults4.tex
\section{Experimental Evaluation} \label{sec:expresults}

In this section, we will present our D1HT and 1h-Calot experimental results, which will be complemented by our analytical evaluations presented in Section \ref{sec:anresults}.

It is worth noting the extensive experimental results we present in this section. First, we used two radically distinct environments, specifically an HPC datacenter and a worldwide dispersed network. Second, our DHT comparison used the largest experimental testbed set up so far, with up to 4,000 peers and 2,000 physical nodes. Finally, we report the first latency comparison among DHTs and a directory server.

\subsection{Methodology} \label{sec:expmethodology}

The D1HT implementation used in our experiments was presented in Section \ref{sec:d1ht-implementation}, which includes only the base D1HT proposal without any extension. In this way, we should better evaluate the unique D1HT contributions, but we expect that our experimental results will reflect a worst case scenario in relation to production grade and better tuned D1HT implementations, which would probably include a number of well known optimizations (e.g., peer re-an\-nounce\-ments, parallel lookups, etc.), even though our implementation has been already thoroughly tested.

Because 1h-Calot was not implemented by its authors, we had to develop a real implementation of that system for our experiments. To allow for a fair comparison, we implemented 1h-Calot after our D1HT code, and both systems share most of the code, in an effort to ensure that differences in the results are not due to implementation issues. Because each 1h-Calot maintenance message carries just one event, it does not make sense to include counters in its message headers, which will then have the format shown in Figure \ref{fig:headers}.

In 1h-Calot each event incurs 2$n$ maintenance messages (including acks), and each peer sends four heartbeats per minute (which are not acknowledged), and so the analytical average 1h-Calot peer maintenance bandwidth will be given by

\vspace{-15pt}
{\begin{equation} \label{equ:income-1h-calot}
{B_{Calot} = ( r \cdot (v_c + v_a) + 4 \cdot n \cdot v_h / 60) ~ \textrm{ bps}
}\end{equation}}
\vspace{-12pt}

\noindent where $v_c$, $v_a$ and $v_h$ are the sizes of the maintenance, acknowledgment and heartbeat messages, respectively (as shown in Figure \ref{fig:headers}).

Each experiment evaluated both systems with a specific session length $S_{avg}$ and a given network size $n$. In all experiments, we used $S_{avg}$=174 min, as this value is representative of Gnutella \cite{saroiu02} and it was used in other DHT studies (e.g., \cite{gupta09,monneratD1HT}). In some experiments, we also used $S_{avg}$=60 min to exercise the systems under more dynamic scenarios. The bandwidth results considered only the traffic for routing table maintenance and peer failure detection, as the other overheads involved, such as lookup traffic and routing table transfers, should be the same for all single-hop DHTs. For all experiments, we defined the routing tables with 6K entries (around 36KB).

Each experiment had two phases, where the first one was used to grow the system up to the target size and the second phase was used for the measurements. In the first phase, each system started with just eight peers, and one peer joined per second until the target size was reached, resulting in a steep growth rate (the systems doubled in size in just eight seconds, with an eightfold growth in less than one minute), which should stress the joining protocols. The second phase always lasted for 30 min, while each peer performed random lookups. We ran each experiment three times and reported the average results.

In both phases of all the experiments, the systems were churned according to Equation \ref{equ:r} and the chosen $S_{avg}$ (60 or 174 min), and the peer leaves were random. Half of the peer leaves were forced with a POSIX \texttt{SIGKILL} signal, which does not allow the leaving peer to warn its neighbors nor to flush any buffered event. To maintain the system size, any leaving peer rejoined the system in three minutes with, unless otherwise stated, the same IP and ID, which allowed us to evaluate both systems in a scenario with concurrent joins and leaves.

Even though our experiments stressed the joining protocols and imposed a high rate of concurrent joins and leaves, both D1HT and 1h-Calot were able to solve more than 99\% of the lookups with a single hop in all experiments.

\subsection{PlanetLab Bandwidth Experiments} \label{sec:exp_planetlab}

\input{fig-exp-bandwidth-PlanetLab.tex}

To evaluate the system overheads in a worldwide dispersed environment, we ran experiments using 200 physical PlanetLab \cite{bavier04} nodes, with either 5 or 10 D1HT and 1h-Calot peers per node, leading to system sizes of 1K or 2K peers, respectively. Each peer performed one random lookup per second during the second phase of our PlanetLab experiments.

Figure \ref{fig:planetlab} shows the sum of the outgoing maintenance bandwidth requirements of all peers for each system. The figure shows that both DHTs had similar overheads for the smaller system size, while with 2K peers the demands of 1h-Calot were 46\% higher than those of D1HT. The more extensive experiments and analyses presented in Sections \ref{sec:exp_hpc_band} and \ref{sec:anresults} will show that this difference will significantly increase with larger system sizes. Figure \ref{fig:planetlab} also shows that the analyses of both systems were able to predict their bandwidth demands, which differ to some extent from previous results \cite{monnerat09}, where the D1HT analysis overestimated its overheads by up to 25\%. We credit those differences to a few factors. First, D1HT had an increase in bandwidth demands due to the implementation of the mechanisms to close $\Theta$ intervals based on Equation \ref{equ:e}, which was not used in the experiments reported in \cite{monnerat09}. Additionally, because the D1HT analysis is strongly dependent on $\rho$=$\lceil\log_2(n)\rceil$, it leads to more precise predictions when $\rho$ is close to $\log_2(n)$ (i.e., when $n$ is slightly smaller than a power of 2, which is the case with all experiments presented here) and to overestimated results when $\rho$ is significantly greater than $\log_2(n)$ (as with the results presented in \cite{monnerat09}).
\input{tabclusters.tex}

\input{fig-exp-bandwidth-HPC.tex}

\subsection{HPC Bandwidth Experiments} \label{sec:exp_hpc_band}

We also performed experiments on a subset of five clusters at a Seismic Processing HPC datacenter \cite{panetta07} (see Table \ref{tab:clusters}). In that network, each node has a Gigabit Ethernet connection to an edge switch, while each edge switch concentrates 16 to 48 nodes and has a 2 Gbps or 10 Gbps Ethernet connection to a non-blocking core switch.

Each peer performed one random lookup per second during the second phase of these experiments, which were conducted with the clusters under normal datacenter production, where typically most of the nodes were experiencing 100\% CPU use, as imposed by the Seismic Processing parallel jobs. Nevertheless, we were able to run all our experiments smoothly, without any single interference in the normal datacenter production, confirming that it is absolutely feasible to run these DHT systems in heavily loaded production environments. In fact, in all our HPC bandwidth experiments the average CPU use per D1HT peer was below 0.1\% (including the cycles used by the joining mechanism and the lookups).

Figures \ref{fig:exp-band-60} and \ref{fig:exp-band-174} show the sum of the outgoing maintenance bandwidth requirements of all peers for each system for different churn rates. We plotted the measured and analytical requirements, showing that, as in the PlanetLab results, the analyses for both systems were precise. The figures also show that D1HT had lower maintenance bandwidth requirements for all cases studied, once more confirming in practice that D1HT can provide a more lightweight DHT implementation than 1h-Calot. The analytical results that will present in Section \ref{sec:anresults} show that the difference in favor of D1HT will grow to more than one order of magnitude for bigger systems.

We also ran our biggest case (4,000 peers) with the leaving peers rejoining with new IPs and IDs to evaluate whether the reuse of IDs caused any relevant bias in our results. In fact, without reusing the IDs, the fraction of the lookups solved with one hop dropped by less than 0.1\%, but it remained well above our 99\% target, which allowed us to conclude that the reuse of IDs did not have any significant effect in our results.

\pagebreak

\subsection{HPC Latency Experiments} \label{sec:exp_hpc_latency}

In this section, we will present our latency experiments performed in the HPC datacenter. As the lookup latencies are sensitive to the network load, we used 400 idle nodes from Cluster A (see Table \ref{tab:clusters}) but, as we expect that DHTs should be able to be used in heavy loaded production environments, we measured the latencies with those nodes both in the idle state and under 100\% CPU load (by running two burnP6 \cite{cpuburn} instances per node). Because we used dedicated nodes, we could increase the lookup rate during the second phase of the experiments to 30 lookups/sec per peer, which allowed us to evaluate the systems under an intense lookup load.

In addition to D1HT and 1h-Calot, we also ran a multi-hop DHT (Pastry) and a directory server (Dserver). In an effort to avoid inserting bias due to implementation issues, Dserver was essentially a D1HT system with just one peer. We first ran Dserver in a Cluster B node, which reached 100\% CPU load when serving lookups from 1,600 peers, thus providing a first indication of the scalability issues of this client/server approach, after which we picked up a dedicated node from Cluster F. For the multi-hop DHT, we used Chimera \cite{chimera} not only because it implements Pastry \cite{pastry01} (using base 4), which is one of the most prominent multi-hop DHTs, but also because it did not require any prerequisites to be installed in the HPC clusters (e.g., Java, Python, etc.).

As our time windows with the dedicated cluster were limited, we ran the four systems concurrently in each experiment. To study different system sizes, we varied the numbers of DHT peers and Dserver clients per node from two to ten. For example, when running six peers per node, we concurrently ran six D1HT peers, six 1h-Calot peers, six Chimera peers and six Dserver clients in each cluster A node.

The D1HT and 1h-Calot peers were churned with $S_{avg}$=174 min, while Dserver and Chimera were not churned. To verify whether the base latencies of the studied systems differ due to implementation issues, we first ran the four systems with just two peers and the observed one-hop latencies were quite similar (around 0.14 ms).

\input{fig-exp-latency.tex}

Figure \ref{fig:exp-latency-burn0} and \ref{fig:exp-latency-burn2} show the latencies measured with the nodes in the idle state and under 100\% CPU load, respectively. As the measured Chimera latencies were higher than expected, we also plotted the expected Chimera latencies assuming that each hop takes 0.14 ms. We believe that the differences between the measured and expected Chimera latencies were due to either implementation issues or measurement artifacts, but even the expected latencies are much higher than those for the single-hop DHTs, which confirms that a multi-hop DHT solution is less suitable for latency-sensitive applications. While Chimera latencies could be improved by using a larger base (e.g., 16), its performance would still be worse than that of D1HT.

We can see from Figure \ref{fig:exp-latency-burn0} that all systems, except for Chimera, presented very similar latencies with idle nodes and smaller system sizes, which was expected because D1HT and 1h-Calot solved more than 99\% of the lookups with one hop, while Dserver ran similar code. However, Dserver started to lag behind the single-hop DHT systems at 3,200 peers (120\% higher latencies), and at 4,000 peers it provided latencies more than one order of magnitude higher, revealing its scalability limitations.

\input{fig-exp-latency-200x400nodes.tex}

We may observe in Figure \ref{fig:exp-latency-burn2} that the latencies of all systems degraded with busy nodes and that, quite surprisingly, the D1HT and 1h-Calot latencies increased slightly with the system size when running on busy nodes. To verify whether this unexpected behavior was related to the extra load generated by the artifact of running several peers and four different systems per node under an intense lookup rate and full CPU load, we ran the 100\% CPU load experiments with just 200 physical nodes, varying again the number of peers per node from two to ten. The results are plotted in Figure \ref{fig:exp-latency-200x400} along with the latencies measured with 400 nodes. For simplicity, in Figure \ref{fig:exp-latency-200x400}, we only plot the D1HT results, even though we also ran 1h-Calot, Chimera and Dserver for both the 200 and 400 node experiments. Confirming our hypothesis, the figure indicates that the latency degradation observed was related to the number of peers per physical node (and the overload they imposed on the already 100\% busy nodes), as the latencies measured with 200 and 400 nodes and the same number of peers per node were quite similar, even though the 400-node systems had twice the size. For instance, with four peers per node, the average latencies measured with 200 nodes (total of 800 peers) and 400 nodes (total of 1,600 peers) were both 0.15 ms. With eight peers per node, the results with 200 nodes (total of 1,600 peers) and 400 nodes (total of 3,200 peers) were 0.23 ms and 0.24 ms, respectively. These results indicate that the D1HT lookup latencies should not vary with the system size, but they can degrade with overloaded peers, while they are still similar to or better than those provided by Dserver and Chimera.

%% file: fig-exp-bandwidth-PlanetLab.tex
\begin{figure*}[!b]
\vspace{-13pt}
\centering
\includegraphics[width=9cm]{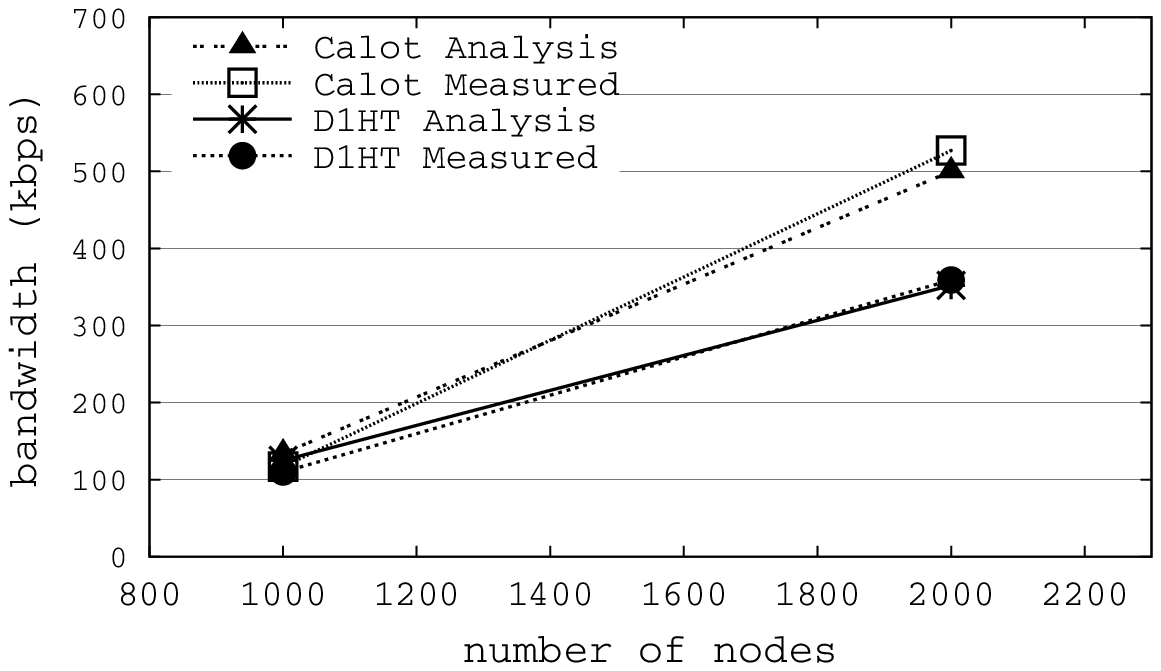}
\caption{Experimental and analytical outgoing maintenance bandwidth demands for D1HT and 1h-Calot in the PlanetLab for $S_{avg}$ = 174 min.}
\label{fig:planetlab}
\vspace{-13pt}
\end{figure*}

%% file: tabclusters.tex
\begin{table*}[!tb]
\centering
\begin{tabular*}{9.17cm}{|c|c|c|c|}
\hline

\textbf{Cluster} & \textbf{\# nodes} & \textbf{CPU} & \textbf{OS}\\

 \hline

 A & 731 &Intel Xeon 3.06GHz single core&Linux 2.6\\
 \hline

 B & 924 &AMD Opteron 270 dual core&Linux 2.6\\
 \hline

 C & 128 &AMD Opteron 244 dual core&Linux 2.6\\
 \hline

 D & 99 &AMD Opteron 250 dual core&Linux 2.6\\
 \hline

 F & 509 &Intel Xeon E5470 quad core&Linux 2.6\\
 \hline

\end{tabular*}
\caption{Clusters used in our experiments. Each node has two CPUs.}
\label{tab:clusters}
\vspace{-20pt}
\end{table*}

%% file: fig-exp-bandwidth-HPC.tex
\begin{figure*}[!b]
\vspace{-14pt}
\centering
\hspace{-62pt}
\subfigure[$S_{avg}$ = 60 min.]{
  \label{fig:exp-band-60}
  \includegraphics[width=7.40cm]{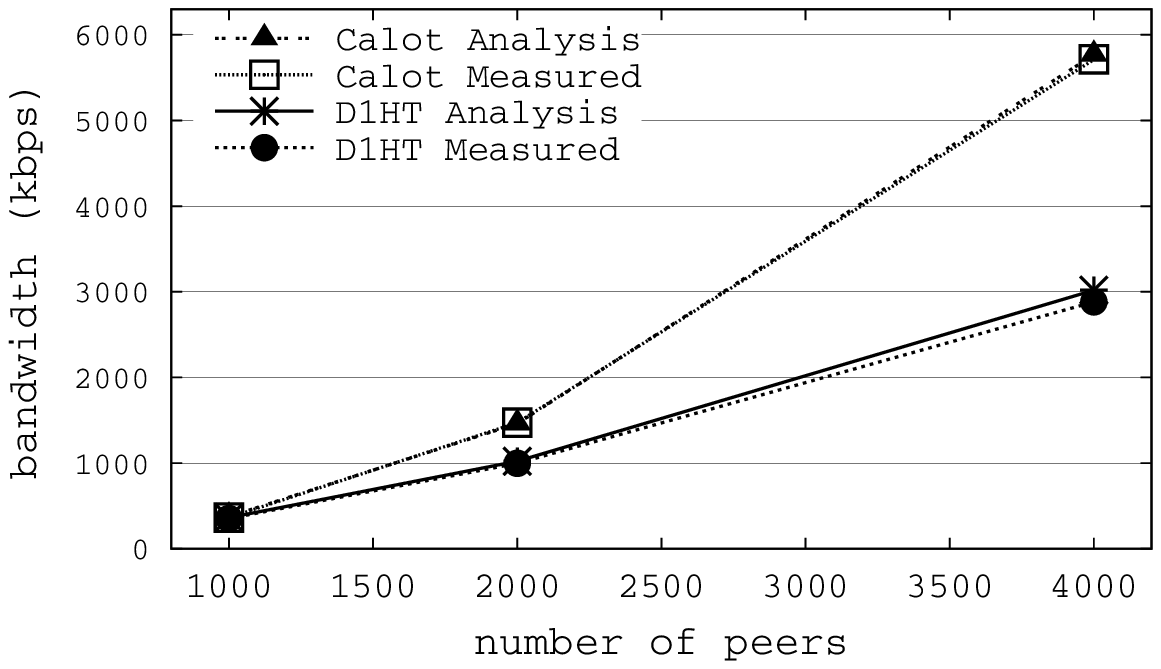}
  }
\subfigure[$S_{avg}$ = 174 min.]{
  \label{fig:exp-band-174}
  \includegraphics[width=7.30cm]{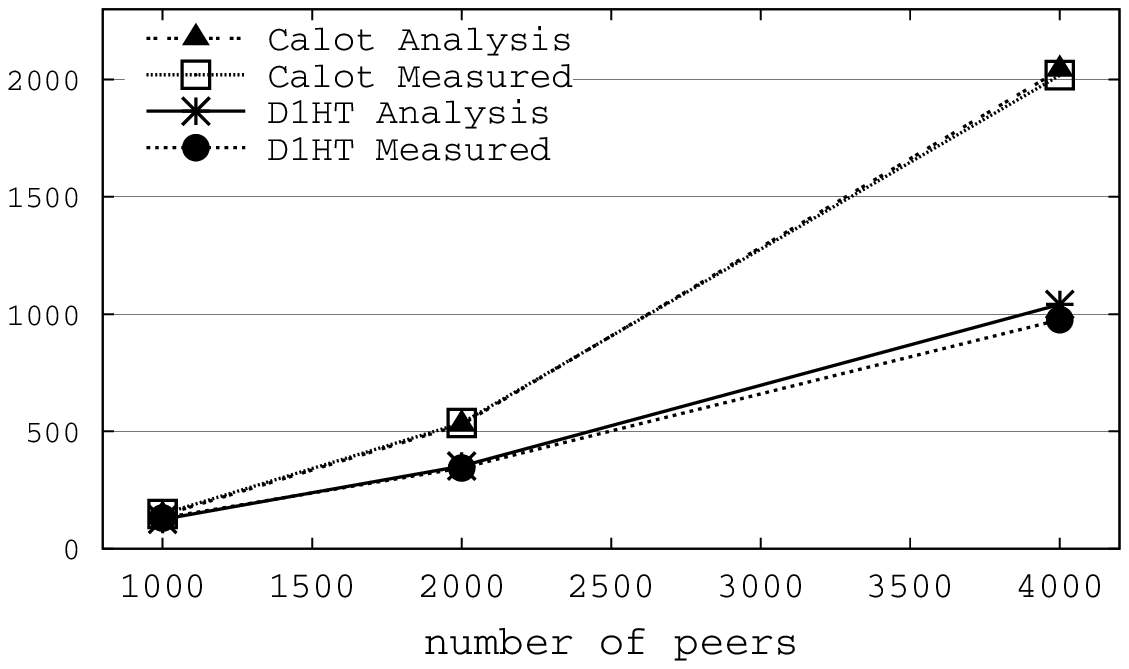}
  }
\hspace{-57pt}
\caption{Experimental and analytical outgoing maintenance bandwidth demands for D1HT and 1h-Calot in the HPC datacenter.}
\label{fig:exp-band-hpc}
\end{figure*}

%% file: fig-exp-latency.tex
\begin{figure*}[!t]
\centering
\hspace{-62pt}
\subfigure[Idle nodes.]{
  \label{fig:exp-latency-burn0}
  \includegraphics[width=7.40cm]{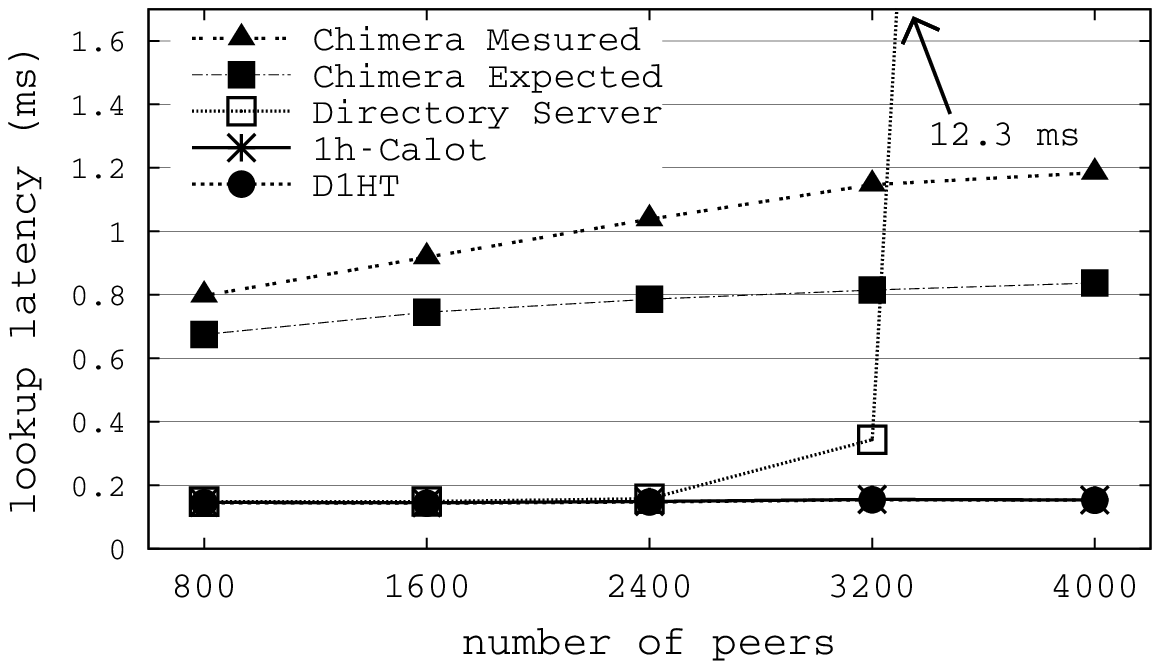}
  }
\subfigure[Nodes under 100\% CPU use.]{
  \label{fig:exp-latency-burn2}
  \includegraphics[width=7.30cm]{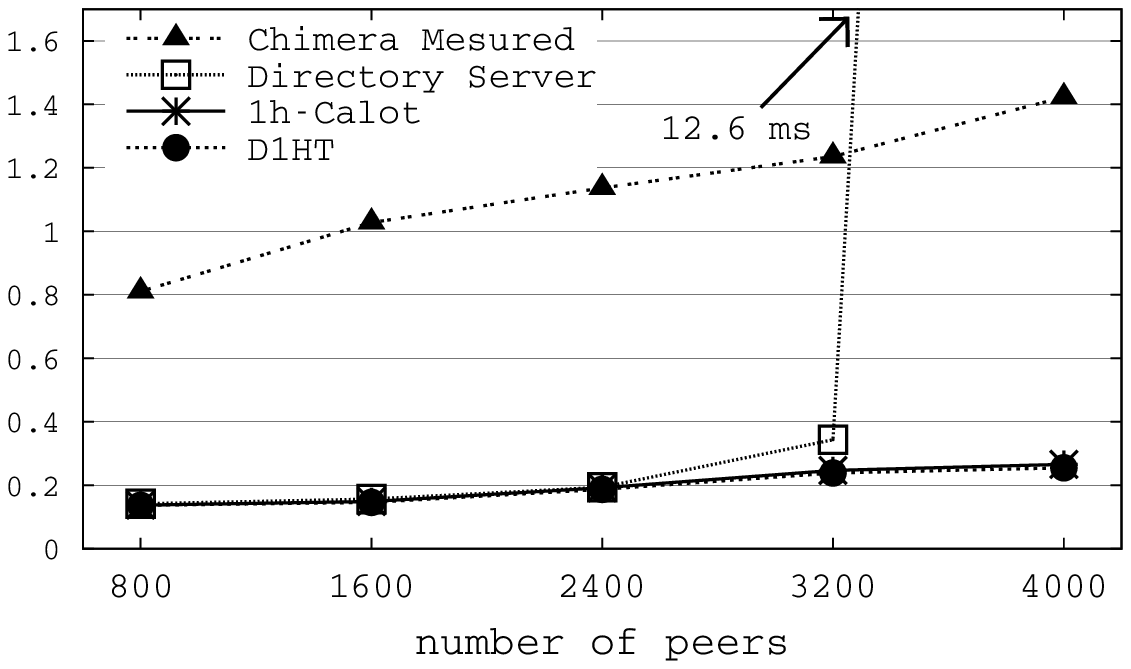}
  }
\hspace{-57pt}
\caption{Lookup latencies measured in the HPC environment with idle and busy (100\% CPU use) nodes.}
\label{fig:exp-latency}
\end{figure*}

%% file: fig-exp-latency-200x400nodes.tex
\begin{figure*}[!t]
\centering
\includegraphics[width=10cm]{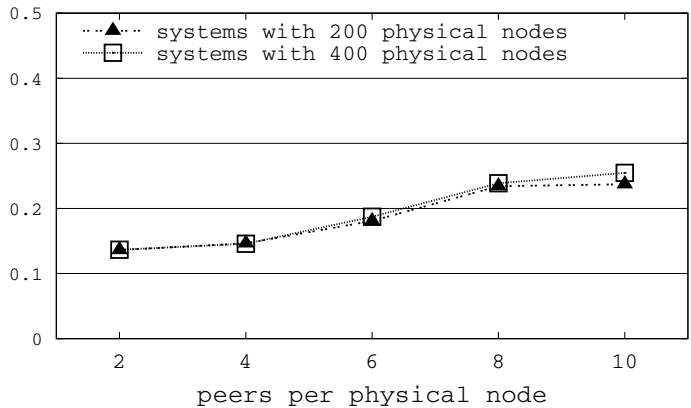}
\caption{Lookup latencies measured in the HPC environment with busy (100\% CPU use) nodes.}
 \label{fig:exp-latency-200x400}
\end{figure*}

%% file: 70-anresults3.tex
\section{Analytical Results} \label{sec:anresults}

As our experiments have validated the 1h-Calot and D1HT analyses, and the OneHop analysis had already been validated in a previous work \cite{gupta09}, we will now compare those three systems analytically. As discussed in Section \ref{sec:relatedwork}, the 1h-Calot results presented in this section should also be valid for the 1HS \cite{risson06a} and SFDHT \cite{sfdht09} systems. In a previous work \cite{monneratD1HT}, we have already provided an extended D1HT analysis, studying the variation of the D1HT overheads and $\Theta $ intervals for different values of $f$, churn rates and system sizes; thus, here we will focus on comparing the overheads of the systems being studied.

As in our experiments, our analytical results compute only the traffic for routing table maintenance, we used $f$=1\%, and we assumed random events and lookups. The OneHop analysis is available from \cite{gupta09}, for which we will consider the same message formats used in our D1HT implementation, as shown in Figure \ref{fig:headers}, because they have been shown to be realistic in practice. Besides, the OneHop results always considered the optimal topological parameters and did not account for the failure of slice and unit leaders. The OneHop and 1h-Calot results do not consider message delays, while for D1HT we used $\delta_{avg}$=0.25 sec, which is an overestimation compared to the Internet delay measurements presented in \cite{saroiu02}.

\input{fig-analysis-bandwidth.tex}

We varied the system size from $10^4$ to $10^7$, which are representative of environments ranging from large corporate datacenters to huge Internet applications, and studied average sessions of 60, 169, 174 and 780 min, where the latter three were observed in KAD \cite{steiner09}, Gnutella \cite{saroiu02} and BitTorrent \cite{andrade05} studies. This range of session lengths is more comprehensive than those used in most DHT evaluations (e.g., \cite{gupta09,li04comparing,monneratD1HT,monnerat09,accordion:nsdi05,superpeers03}) and is representative of widely deployed P2P applications.

\pagebreak
\input{fig-quarantine.tex} 
Figures \ref{fig:an-band-60} to \ref{fig:an-band-780} show log-log plots comparing the analytical bandwidth demands of D1HT and 1h-Calot peers against those of the best (ordinary nodes) and worst (slice leaders) OneHop cases. From these figures we can see that the OneHop hierarchical approach imposes high levels of load imbalance between slice leaders and ordinary nodes. Moreover, a D1HT peer typically has maintenance requirements one order of magnitude smaller than OneHop slice leaders, while attaining similar overheads compared to ordinary nodes. Compared to D1HT, the 1h-Calot overheads were at least twice greater and typically one order of magnitude higher for the cases studied. The requirements for a D1HT peer in systems with $n$=$10^6$ and average sessions of 60, 169, 174 and 780 min are 20.7 kbps, 7.3 kbps, 7.1 kbps and 1.6 kbps, respectively. In contrast, the overheads for the OneHop slice leaders and 1h-Calot peers for systems with $n$=$10^6$ and KAD dynamics were above 140 kbps.

The Quarantine analysis will be based on data from studies that observed that $31\%$ of the Gnutella sessions \cite{chu02} and $24\%$ of the KAD sessions \cite{steiner09} lasted less than 10 minutes, which is a convenient value for the Quarantine period $T_q$. Then, Figures \ref{fig:quarantine-169} and \ref{fig:quarantine-174} show the overhead reductions provided by Quarantine for D1HT systems with dynamics similar to KAD and Gnutella, with $T_q$=10 min. We can see that the maintenance bandwidth reduction grows with the system size, as for very small systems the overheads were dominated by messages with TTL=0, which are always sent even when there are no events to report. Although the length of the Quarantine period studied was less than 6\% of the average session length for both systems, the overhead reductions with $n$=$10^7$ for KAD and Gnutella dynamics reached 24\% and 31\% respectively, showing the effectiveness of the Quarantine mechanism.

%% file: fig-analysis-bandwidth.tex
\begin{figure*}[!t]
\centering
\hspace{-10pt}
\subfigure[$S_{avg}$=60 min.]{
  \label{fig:an-band-60}
  \includegraphics[width=7.15cm]{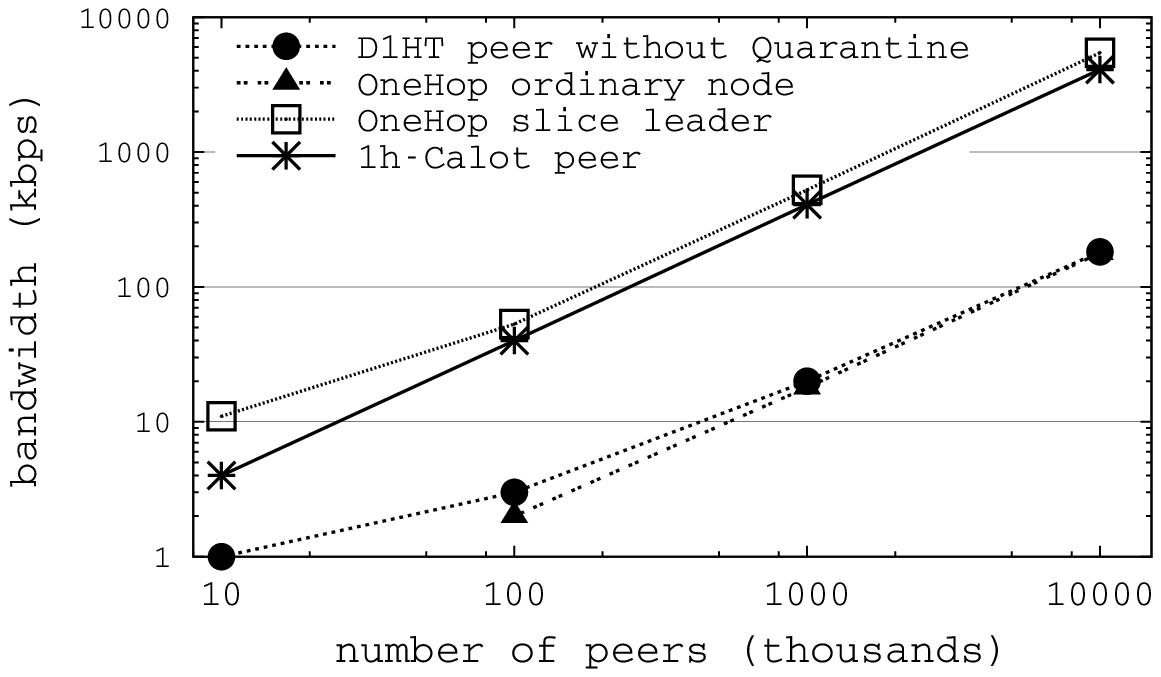}
  }
\hspace{-5pt}
\subfigure[$S_{avg}$=169 min (KAD dynamics).]{
  \label{fig:an-band-169}
  \includegraphics[width=7.1cm]{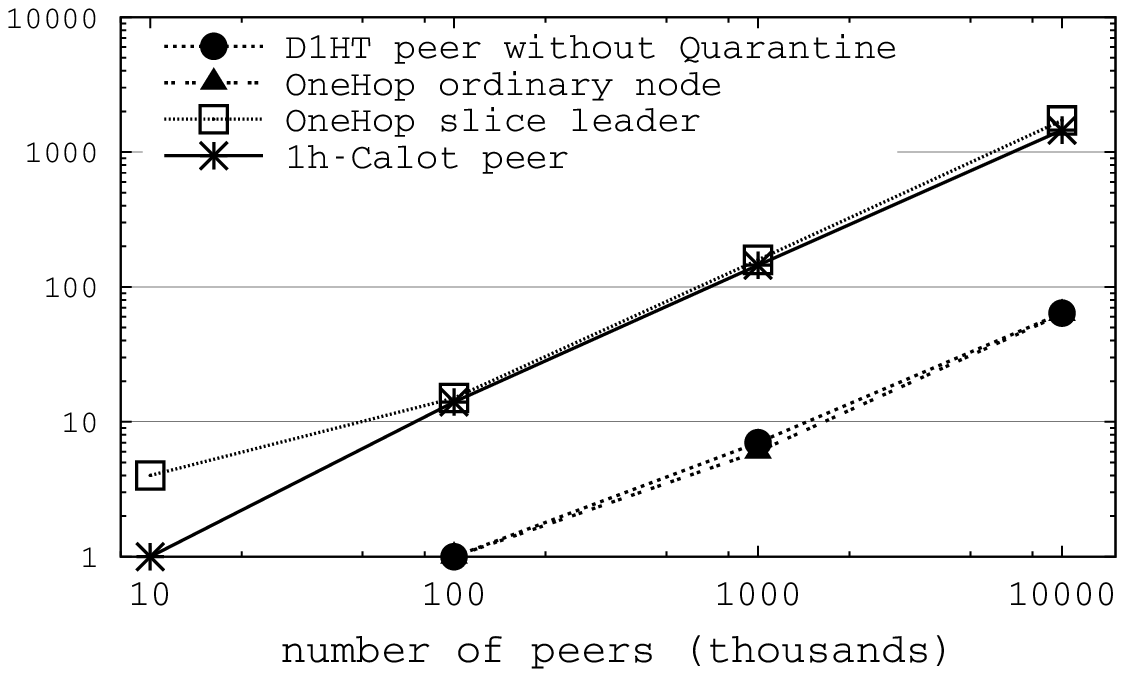}
  }
\\
\hspace{-10pt}
\subfigure[$S_{avg}$=174 min (Gnutella dynamics).]{
  \label{fig:an-band-174}
  \includegraphics[width=7.15cm]{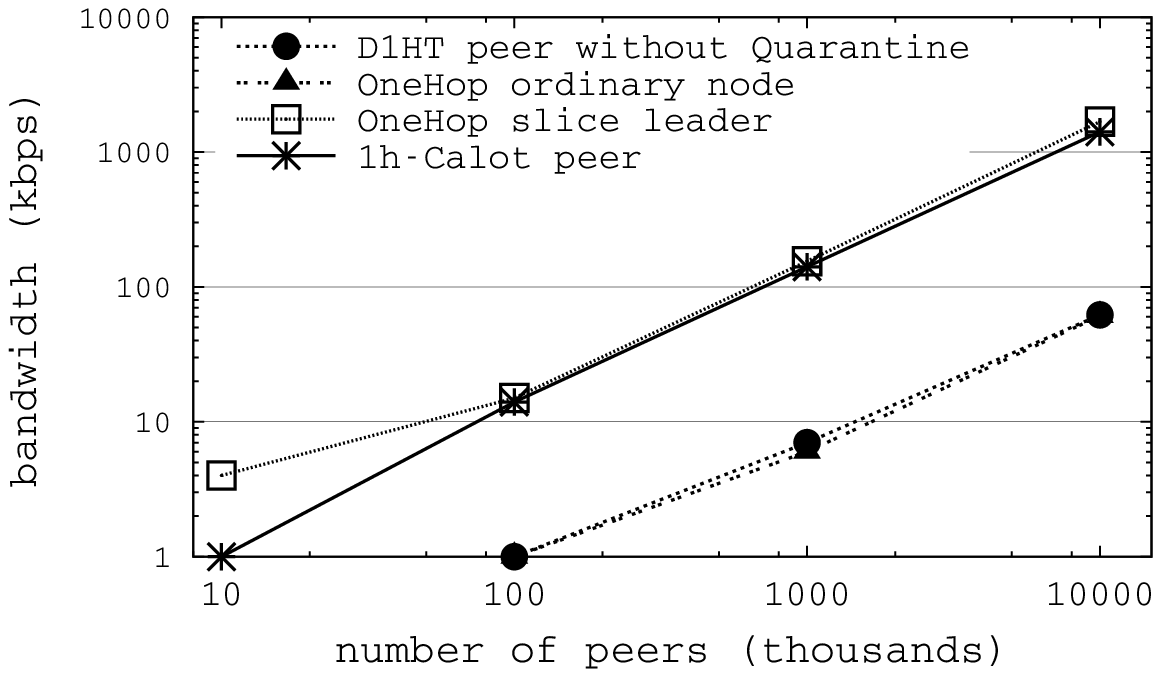}
  }
\hspace{-5pt}
\subfigure[$S_{avg}$=780 min (BitTorrent dynamics).]{
  \label{fig:an-band-780}
  \includegraphics[width=7.1cm]{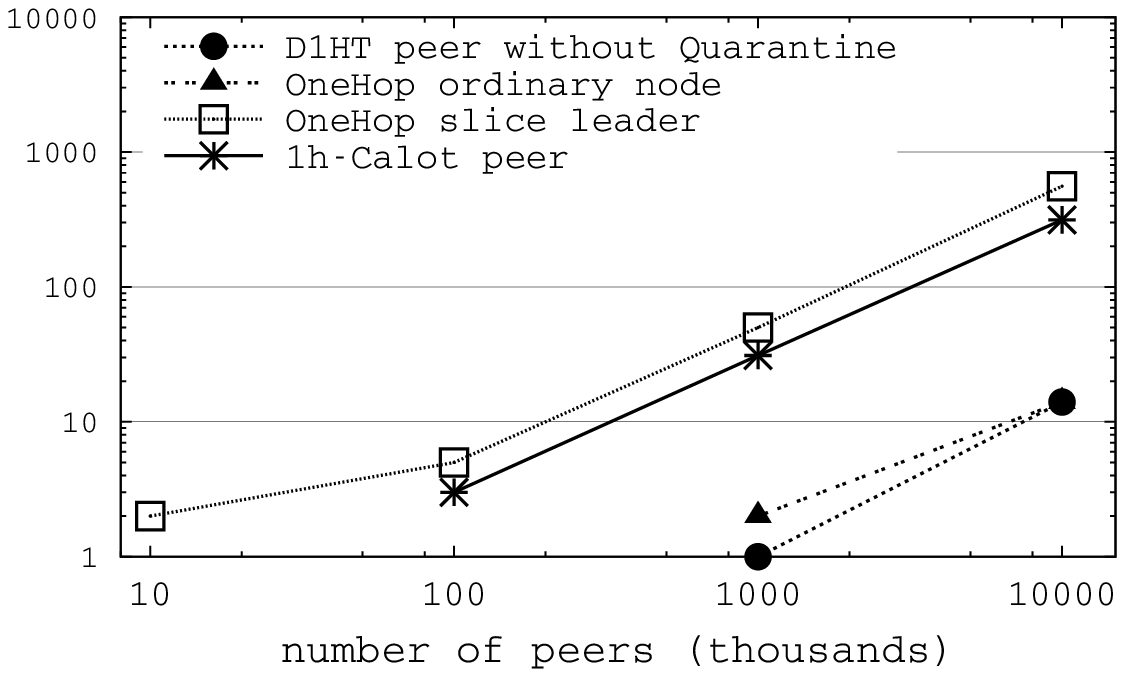}
  }
\vspace{-8pt}
\caption{Log-log plots showing the analytical outgoing maintenance bandwidth demands for D1HT, 1h-Calot and OneHop (we do not show values below 1 kbps).}
\label{fig:an-band}
\vspace{-18pt}
\end{figure*}

%% file: fig-quarantine.tex
\begin{figure*}[!b]
\vspace{-18pt}
\centering
\hspace{-50pt}
\subfigure[Quarantine gains with KAD dynamics ($q$=0.76$n$).]{
  \label{fig:quarantine-169}
  \includegraphics[width=7.0cm]{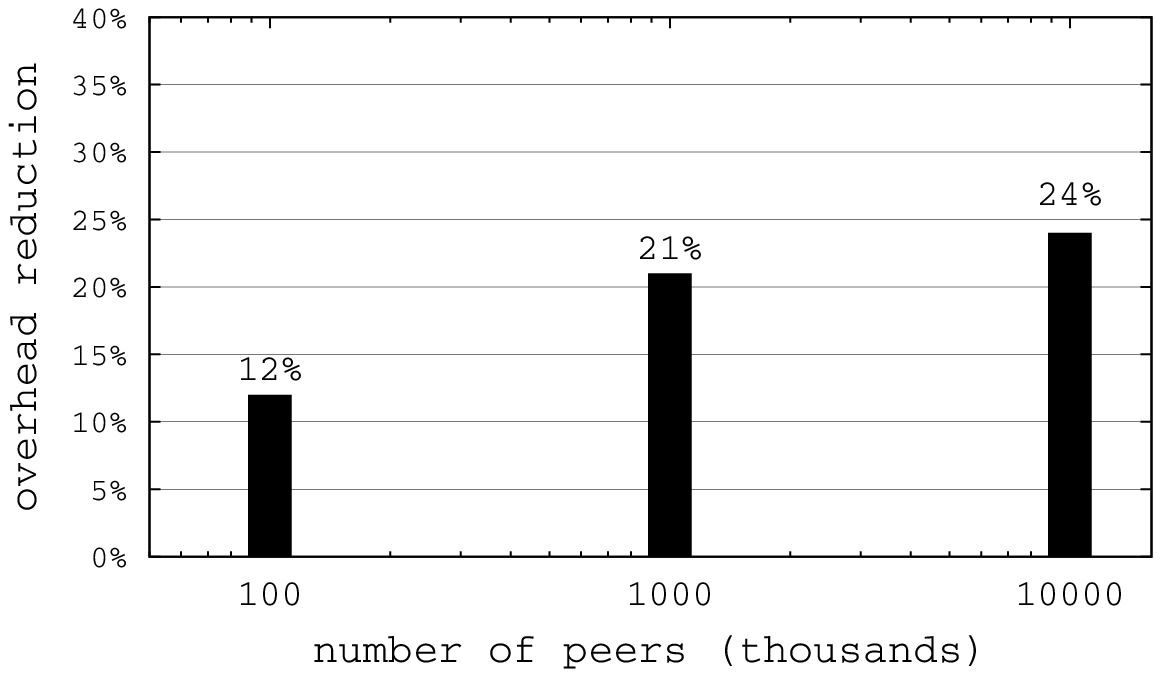}
  }
\subfigure[Quarantine gains with Gnutella dynamics ($q$=0.69$n$).]{
  \label{fig:quarantine-174}
  \includegraphics[width=6.9cm]{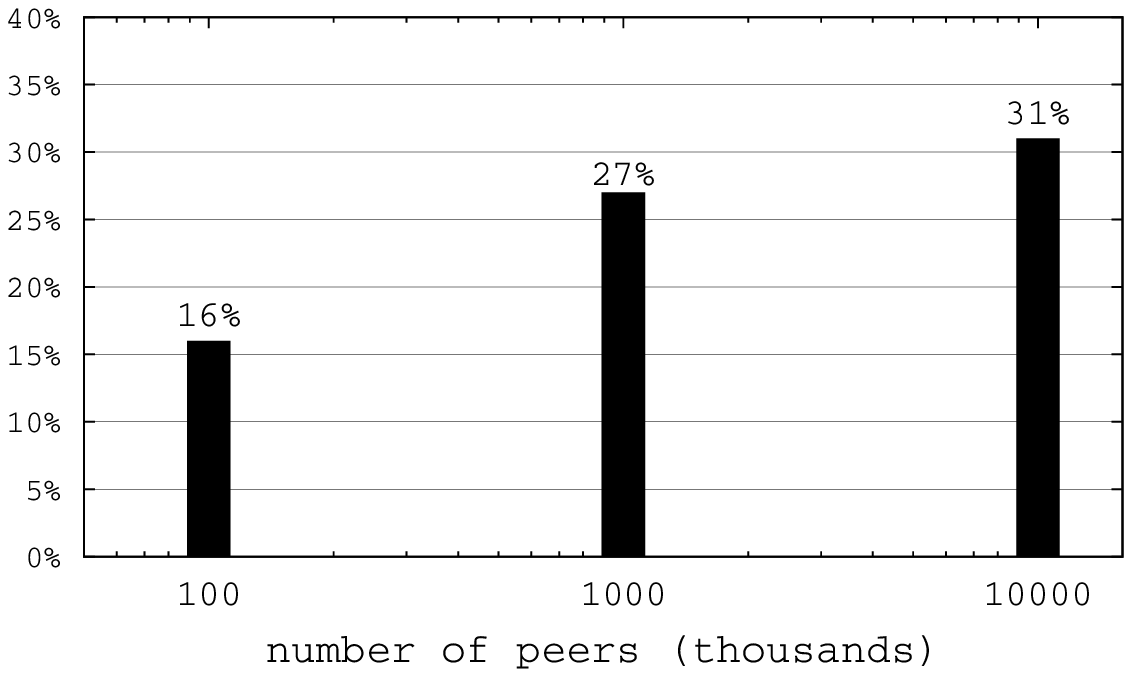}
  }
\hspace{-32pt}
\caption{Estimated overhead reductions brought by Quarantine for systems with KAD and Gnutella behaviors.}
\label{fig:quarantine}
\vspace{-8pt}
\end{figure*}

%% file: 80-discussion4-cla.tex
\section{Discussion} \label{sec:discussion}

In addition to validating the D1HT analysis, our experiments confirmed that D1HT was able to solve more than 99\% of the lookups with a single hop and very low CPU and memory overhead, even with nodes under high CPU load or peers widely dispersed over the Internet. For instance, in all our HPC bandwidth experiments, the average CPU usage per peer was less than 0.1\%, and the memory sizes for routing table storage were around 36 KB per peer.

Our results also showed that D1HT had the lowest overheads among all single-hop DHTs that support dynamic environments, with typical reductions of one order of magnitude for big systems. D1HT´s performance advantage was due to its ability to group events for dissemination with a pure P2P approach, even for large and dynamic environments where the system size and peer behaviors change over time. In contrast, other single-hop DHTs either do not provide means for their peers to group events \cite{tang05,sfdht09,risson06a,risson09,rodrigues02} or use a hierarchical approach with high levels of load imbalance and other intrinsic issues \cite{gupta09}.

Compared to a directory server, D1HT achieved similar latencies for small systems while attaining better scalability, which allowed it to provide latencies up to one order of magnitude better for the larger systems studied, even with nodes under full CPU load, revealing that D1HT is also an attractive solution for large-scale latency-sensitive applications.

Considering that back in 2004 the BitTorrent peer average download speed was already around 240 kbps \cite{pouwelse05}, we may assume that the D1HT with 1.6-16 kbps maintenance overheads should be negligible for systems with one to ten million peers and BitTorrent behavior. Moreover, as other studies found that most domestic connections have at \emph{least} 512 kbps of downstream bandwidth with very low occupation \cite{dischinger07,maier09}, we argue that we should not penalize lookup latencies to save fractions below 10\% of the available bandwidth. Thus, in the near future, even systems with up to ten million nodes with KAD or Gnutella dynamics will probably be able to benefit from the lowest latencies provided by D1HT with less then 65 kbps maintenance overheads.

While 1h-Calot could also be used in HPC and ISP datacenters, its use would require the development and maintenance of a DHT dedicated to those environments. In contrast, the distinguished D1HT ability to provide both low latency and small overheads may allow it to support a wide range of environments, in such a way that D1HT can act as a commodity DHT, which makes D1HT a very attractive option for these corporate datacenters, specially as they are preferably built on commodity hardware and software \cite{top500,Barroso2008}.

%% file: 300-conclusion7.tex
\section{Conclusion} \label{sec:conclusion}

While latency issues should become much more critical than bandwidth restrictions over time, the first DHT proposals have opted to trade off latency for bandwidth, and recent single-hop DHTs typically have either high overheads or poor load balance. In this work, we presented D1HT, which has a pure P2P and self-organizing approach and is the first single-hop DHT combining low maintenance bandwidth demands and good load balance, along with a Quarantine mechanism that is able to reduce the overheads caused by volatile peers in P2P systems.

We performed a very extensive and representative set of DHT comparative experiments, which validated the D1HT analysis and was complemented by analytical studies. Specifically, by using an experimental environment that was at least 10 times greater than those of all previous DHT comparative experiments, the present work became the first to assess five key aspects of DHT behavior in such practical settings. Concretely,  the present  work is the first to i) report DHT comparative experiments in two different environments; ii) compare DHT lookup latencies; iii) perform experiments with two different single-hop DHTs; iv) compare the latencies of multi and single hop DHTs; and v) compare DHTs to central directories.

Overall, our results showed that D1HT consistently had the lowest maintenance costs among the single-hop DHTs, with overhead reductions of up to one order of magnitude for large systems, and indicated that D1HT could be used even for huge systems with one million peers and dynamics similar to those of popular P2P applications.

Our experiments also showed that D1HT provides latencies comparable to those of a directory server for small systems, while exhibiting better scalability for larger ones, which shows that it is an attractive and highly scalable option for very large latency-sensitive environments.

We believe that D1HT may be very useful for several Internet and datacenter distributed applications, since the improvements in both bandwidth availability and processing capacity that we should continuously get will bring performance expectations to users and applications, which can be frustrated by latency constraints. In addition, trends in High Performance Computing, ISP and Cloud Computing environments indicate significant increases in the system sizes, which will challenge the scalability and fault tolerance of client/servers solutions.

As a consequence of our extensive set of results, we may conclude that D1HT can potentially be used in a multitude of environments, ranging from HPC and ISP datacenters to huge P2P applications deployed over the Internet, and that its attractiveness should increase over time. This ability to support such a wide range of environments may allow D1HT to be used as an inexpensive and scalable commodity software substrate for distributed applications. As one step in that direction, we have made our D1HT source code available for free use \cite{D1HTsource}.

\vspace{-6pt}
\section*{Acknowledgments}
{
We would like to thank Petrobras for providing access to the clusters used in the experiments and authorizing the public dissemination of the results. The PlanetLab experiments would not have been possible without the support of the PlanetLab and Brazilian RNP teams. This research was partially sponsored by Brazilian CNPq and FINEP.}

%% file: 600-end.tex
\bibliographystyle{plain}
\bibliography{bibedit2}
\end{document}